\documentclass[11pt,letterpaper]{amsart}
\usepackage{amsmath,amssymb,amsthm,amsfonts,mathscinet,eucal}
\usepackage{accents}
\usepackage{enumerate}
\usepackage{multirow}
\usepackage{color}
\usepackage{gnuplot-lua-tikz}
\usepackage{cite, hyperref}
\usepackage[noabbrev,capitalize]{cleveref}

\DeclareFontFamily{U}{mathx}{\hyphenchar\font45}
\DeclareFontShape{U}{mathx}{m}{n}{
      <5> <6> <7> <8> <9> <10>
      <10.95> <12> <14.4> <17.28> <20.74> <24.88>
      mathx10
      }{}
\DeclareSymbolFont{mathx}{U}{mathx}{m}{n}
\DeclareFontSubstitution{U}{mathx}{m}{n}
\DeclareMathAccent{\widecheck}{0}{mathx}{"71}

\renewcommand{\rho}{\varrho}
\renewcommand{\phi}{\varphi}
\DeclareMathOperator{\ADF}{ADF}
\DeclareMathOperator{\PADF}{PADF}
\DeclareMathOperator{\CDF}{CDF}
\DeclareMathOperator{\PCDF}{PCDF}
\DeclareMathOperator{\PCC}{PCC}
\DeclareMathOperator{\PSL}{PSL}
\DeclareMathOperator{\GUC}{GUC}
\DeclareMathOperator{\SDC}{SDC}
\DeclareMathOperator{\AC}{C}
\DeclareMathOperator{\PC}{PC}
\DeclareMathOperator{\len}{len}
\DeclareMathOperator{\per}{per}
\DeclareMathOperator{\supp}{supp}

\newcommand{\C}{{\mathbb C}}
\newcommand{\F}{{\mathbb F}}
\newcommand{\N}{{\mathbb N}}
\newcommand{\R}{{\mathbb R}}
\newcommand{\Z}{{\mathbb Z}}
\newcommand{\Zell}{{\Z/\ell\Z}}
\newcommand{\Zn}{{\Z/n\Z}}
\newcommand{\Fp}{\F_p}
\newcommand{\Fpu}{\F_p^*}
\newcommand{\achars}{\widehat{\F_p}}
\newcommand{\mchars}{\widehat{\F_p^*}}
\newcommand{\Fpn}{\F_p^{*n}}
\newcommand{\calA}{{\mathcal A}}
\newcommand{\calF}{{\mathcal F}}

\newcommand{\cyc}[2]{({#1},{#2})_{n,p}}
\newcommand{\conj}[1]{\overline{#1}}
\newcommand{\sums}[1]{\sum_{\substack{#1}}}
\newcommand{\maxs}[1]{\max_{\substack{#1}}}
\newcommand{\floor}[1]{\lfloor{#1}\rfloor}
\newcommand{\ceil}[1]{\lceil{#1}\rceil}
\newcommand{\ft}[1]{\widehat{#1}}
\newcommand{\ift}[1]{\widecheck{#1}}
\newcommand{\uni}[1]{\widetilde{#1}}
\newcommand{\norm}[2]{\|#2\|_{#1}}
\newcommand{\norminf}[1]{\norm{\infty}{#1}}
\newcommand{\normone}[1]{\norm{1}{#1}}
\newcommand{\normonesq}[1]{\norm{1}{#1}^2}
\newcommand{\normtwo}[1]{\norm{2}{#1}}
\newcommand{\normtwosq}[1]{\normtwo{#1}^2}
\newcommand{\Norm}[2]{\left\|#2\right\|_{#1}}

\newcommand{\Normtwo}[1]{\Norm{2}{#1}}
\newcommand{\Normtwosq}[1]{\Normtwo{#1}^2}
\newcommand{\ip}[2]{\langle #1,#2 \rangle}
\newcommand{\quotring}[1]{\C[z]/(z^{#1}-1)}
\newtheorem{theorem}{Theorem}[section]
\newtheorem{proposition}[theorem]{Proposition}
\newtheorem{lemma}[theorem]{Lemma}
\newtheorem{corollary}[theorem]{Corollary}
\theoremstyle{definition}
\newtheorem{remark}[theorem]{Remark}

\title[Sets of Low Correlation Sequences from Cyclotomy]{Sets of Low Correlation Sequences \\ from Cyclotomy}
\author[Castello, Katz, King, and Olavarrieta]{Jonathan M.~Castello, Daniel J.~Katz, \\ Jacob M.~King, and Alain Olavarrieta}
\thanks{This paper is based upon work of Jonathan Castello and  Daniel J. Katz supported in part by the National Science Foundation under Grants DMS-1500856 (for Castello and Katz) and  CCF-1815487 (for Katz).  The work of Alain Olavarrieta on this paper was supported in part by the Graduate Fellowship for Outstanding Research Promise in Science and Mathematics from the College of Science and Mathematics at California State University, Northridge.}
\thanks{Jonathan Castello was with the Department of Mathematics, California State University, Northridge.}
\thanks{Daniel J.~Katz is with the Department of Mathematics, California State University, Northridge.}
\thanks{Jacob King was with the Department of Mathematics, California State University, Northridge and is now with the Data Science Program, George Washington University, Washington, D.C}
\thanks{Alain Olavarrieta was with the Department of Mathematics, California State University, Northridge and is now with the Department of Mathematics, Las Positas College, Livermore, CA}
\date{29 December 2021}

\begin{document}

\begin{abstract}
Low correlation (finite length) sequences are used in communications and remote sensing.
One seeks codebooks of sequences in which each sequence has low aperiodic autocorrelation at all nonzero shifts, and each pair of distinct sequences has low aperiodic crosscorrelation at all shifts.
An overall criterion of codebook quality is the demerit factor, which normalizes all sequences to unit Euclidean norm, sums the squared magnitudes of all the correlations between every pair of sequences in the codebook (including sequences with themselves to cover autocorrelations), and divides by the square of the number of sequences in the codebook.
This demerit factor is expected to be $1+1/N-1/(\ell N)$ for a codebook of $N$ randomly selected binary sequences of length $\ell$, but we want demerit factors much closer to the absolute minimum value of $1$.
For each $N$ such that there is an $N\times N$ Hadamard matrix, we use cyclotomy to construct an infinite family of codebooks of binary sequences, in which each codebook has $N-1$ sequences of length $p$, where $p$ runs through the primes with $N\mid p-1$.
As $p$ tends to infinity, the demerit factor of the codebooks tends to $1+1/(6(N-1))$, and the maximum magnitude of the undesirable correlations (crosscorrelations between distinct sequences and off-peak autocorrelations) is less than a small constant times $\sqrt{p}\log(p)$.
This construction also generalizes to nonbinary sequences.
\end{abstract}

\maketitle

\section{Introduction}

\subsection{Sequences and codebooks}\label{Alice}

In this paper, $\N$ always denotes the set of nonnegative integers and $\Z^+$ the set of strictly positive integers.
An {\it (aperiodic) sequence} is a function $f\colon \Z \to \C$ such that $f(j)\not=0$ for only finitely many $j \in \Z$.
We usually write $f_j$ instead of $f(j)$, and we think of $f$ as a doubly infinite sequence $(\ldots,f_{-1},f_0,f_1,\ldots)$.
The set of aperiodic sequences is well known to be a $\C$-inner product space under the usual componentwise addition, $\C$-scalar multiplication, and dot product.

If $f$ is a sequence, then the {\it support of $f$}, written $\supp(f)$, is the set $\{j \in \Z: f_j\not=0\}$, and if $S$ is a subset of $\Z$, then we say that {\it $f$ is supported on $S$} to mean $\supp(f)\subseteq S$.
The {\it length of $f$}, written $\len(f)$, is the size of the smallest set of consecutive integers containing $\supp(f)$.
Thus, $\len(0)=0$, and if $f\not=0$, then $\len(f)=\max\supp(f)-\min\supp(f)+1$.

A {\it unimodular sequence} is a sequence $f$ such that $\supp(f)$ is a set of consecutive integers, and $|f_j|=1$ for each $j \in \supp(f)$.
Thus, the sequence $(\ldots,0,0,1,\exp(2\pi i/5),-1,\exp(4\pi i/3),0,0,\ldots)$ is unimodular, but neither $(\ldots,0,0,1,2,-1,1,0,0,\ldots)$ nor $(\ldots,0,0,1,0,-1,1,0,0,\ldots)$ is.
If $m \in \Z^+$, then the {\it $m$-ary alphabet} is the set of complex $m$th roots of unity, $\{\exp(2\pi i k/m)\colon 0 \leq k < m\}$, so that the binary alphabet is $\{1,-1\}$.
An {\it $m$-ary sequence} is a unimodular sequence $f$ such that $f_j$ is in the $m$-ary alphabet for each $j \in \supp(f)$.

An {\it (aperiodic) codebook} is a finite set $F$ of aperiodic sequences.
Any property (e.g., length, unimodularity, $m$-arity) that can be predicated of a sequence is also predicated of a codebook if and only if every sequence in the codebook has that property.
So, for example, we can say that $F$ is a binary codebook with $\len(F)=5$ if and only if every sequence in $F$ is a binary sequence of length $5$.
A {\it uniform-length codebook} is a codebook in which all sequences have the same length (so that the codebook has a defined length).

\subsection{Correlation}

In this paper, we are concerned with the correlation of aperiodic sequences.
If $f$ and $g$ are aperiodic sequences and $s\in\Z$, then the {\it (aperiodic) crosscorrelation of $f$ with $g$ at shift $s$} is defined to be
\[
\AC_{f,g}(s) = \sum_{j \in \Z} f_{j+s} \conj{g_j},
\]
where we note that the finite supports of $f$ and $g$ make this sum have only finitely many nonzero terms.
If $f$ and $g$ are supported on $\{0,1,\ldots,\ell-1\}$ for some $\ell\in\N$, then $\AC_{f,g}(s)=0$ when $|s| \geq \ell$.
Also note that $\AC_{g,f}(s)=\conj{\AC_{f,g}(-s)}$.
The {\it (aperiodic) autocorrelation of $f$ at shift $s$} is the crosscorrelation of $f$ with itself at shift $s$, that is, $\AC_{f,f}(s)$.
The autocorrelation of a sequence $f$ at shift zero (i.e., $\AC_{f,f}(0)$) is the squared Euclidean norm (sum of the squared magnitudes of the terms) of $f$.
When $f$ is unimodular, this makes $\AC_{f,f}(0)=\len(f)$.
If $f\not=0$, then the {\it normalization of $f$} is the sequence obtained by scaling $f$ by $1/\sqrt{\AC_{f,f}(0)}$.

Communications and remote sensing systems require sequences whose autocorrelations at every nonzero shift are much smaller in magnitude than the autocorrelation at shift $0$; this ensures accurate timing.
When more than one sequence is in use (as in a multi-user communications network), it is also important that every pair of distinct sequences have low magnitude crosscorrelations at every shift.

There are two principal ways to measure smallness of correlation: worst case (peak) measures and mean square (demerit factor) measures.
If $f$ is an aperiodic sequence, then the {\it peak sidelobe level of $f$} is
\[
\PSL(f) = \maxs{s \in \Z \\ s\not=0} |\AC_{f,f}(s)|,
\]
and if $g$ is also an aperiodic sequence, then the {\it peak crosscorrelation of $f$ with $g$} is
\[
\PCC(f,g) = \max_{s \in \Z} |\AC_{f,g}(s)|.
\]
Note that since $\AC_{g,f}(s)=\conj{\AC_{f,g}(-s)}$ for all $s\in\Z$, we have $\PCC(g,f)=\PCC(f,g)$.
If $F$ is a codebook, then the {\it greatest undesirable correlation of $F$} is
\[
\GUC(F)=\max_{(f,g,s) \in F\times F\times\Z \smallsetminus \{(f,f,0): f \in F\}} |\AC_{f,g}(s)|.
\]
We want $\GUC(F)$ to be as small as possible compared to the autocorrelations at shift zero for sequences in $F$, so we define the {\it smallest desirable correlation of $F$} to be
\[
\SDC(F)=\min_{f \in F} \AC_{f,f}(0).
\]

Schmidt showed \cite[Theorem 1.1]{Schmidt} that, as $\ell$ tends to infinity, if $f$ is a randomly selected binary sequence of length $\ell$, then the expected value of $\PSL(f)/\sqrt{\ell\log\ell}$ tends to $\sqrt{2}$.
Historically, it has been difficult to demonstrate that families of sequences used in communications and sensing also have peak sidelobe level as low or better than this average.
For example, Sarwate \cite{Sarwate-1984-Upper} showed that a binary m-sequence $f$ of length $\ell$ has $\PSL(f) \leq 1+(2/\pi) \sqrt{\ell+1} \log(4 \ell/\pi)$.
And if $f$ is a sequence obtained by cyclically shifting a Legendre sequence of length $\ell$, then Mauduit and S\'ark\"ozy give a result \cite[Cor.~1]{Mauduit-Sarkozy} which implies that $\PSL(f) < 2+18 \sqrt{\ell} \log(\ell)$; this bound is improved to $\PSL(f) \leq 2+2 \sqrt{\ell} + 4/\pi \sqrt{\ell} \log(4 \ell/\pi)$ in this paper (see \cref{Rebecca}).
Note that the logarithms in these formulas are not under the square root, but these bounds are proved using character sum estimates that may involve a great deal of overestimation.
These examples involve binary sequences $f$ of length $\ell$, where $\ell=\len(f)=\AC_{f,f}(0)$, and we want $\PSL(f)$ to not grow much faster than $\sqrt{\ell}=\sqrt{\AC_{f,f}(0)}$.
In general, if we are working with codebooks $F$ of (not necessarily binary) sequences, the appropriate generalization is to say that we want $\GUC(F)$ to not grow much faster than $\sqrt{\SDC(F)}$.
If $\calF$ is a family of codebooks such that $\{\SDC(F): F \in \calF\}$ is unbounded, then we say that $\calF$ has {\it well regulated growth of $\GUC$} to mean that for every real number $\epsilon>0$, the quantity $\GUC(F)/\SDC(F)^{1/2+\epsilon}$ tends to $0$ as $\SDC(F)$ tends to infinity.

Now let us consider mean square measures of smallness of correlation.
If $f\not=0$, then {\it autocorrelation demerit factor of $f$} is
\[
\ADF(f) = \frac{\sums{s \in \Z \\ s\not=0} |\AC_{f,f}(s)|^2}{\AC_{f,f}(0)^2},
\]
which is the sum of the squared magnitudes of all autocorrelation values for the normalization of $f$.
If $f,g\not=0$, then the {\it crosscorrelation demerit factor of $f$ with $g$} is
\[
\CDF(f,g) = \frac{\sum_{s \in \Z} |\AC_{f,g}(s)|^2}{\AC_{f,f}(0) \AC_{g,g}(0)},
\]
which is the sum of the squared magnitudes of all the crosscorrelation values for the normalizations of $f$ and $g$.
Note that $\ADF(f)=\CDF(f,f)-1$, and that $\CDF(g,f)=\CDF(f,g)$ because $\AC_{g,f}(s)=\conj{\AC_{f,g}(-s)}$ for all $s \in \Z$.
Merit factors are just reciprocals of demerit factors (when they are nonzero): the {\it autocorrelation merit factor of $f$} is $1/\ADF(f)$ and the {\it crosscorrelation merit factor of $f$ with $g$} is $1/\CDF(f,g)$.
We prefer to work with demerit factors, since placing the more complicated terms in the numerator makes analysis more tractable.

If $F$ is a codebook, then the {\it (crosscorrelation) demerit factor of $F$} is
\[
\CDF(F)=\frac{1}{|F|^2}\sum_{f,g \in F} \CDF(f,g),
\]
where we note that we are allowing $f=g$ in the summation, so that the sum includes $\CDF(f,f)=\ADF(f)+1$ for each $f \in F$.
For a nonempty binary codebook $F$ of positive length, it is known \cite[Theorem 1]{Liu-Guan} that $\CDF(F)\geq 1$, with equality if and only if $F$ is a complementary set of (nonzero) sequences (i.e., a codebook $F$ in which $\sum_{f \in F} \AC_{f,f}(s)=0$ for every nonzero $s \in \Z$).
It actually follows from a theorem of Katz and Scharf \cite{Katz-Scharf} that for any nonempty codebook such that $\AC_{f,f}(0)$ has the same value for each $f \in F$, we have $\CDF(F)=1$ if and only if $F$ is a complementary set\footnote{Katz and Scharf show that if $C$ is nonzero and $F$ is a nonempty codebook such that $\AC_{f,f}(0)=C$ for every $f\in F$, then $\CDF(F)=1+|F|^{-2} C^{-2} \sum_{s\in \Z \smallsetminus\{0\}} |\sum_{f \in F} \AC_{f,f}(s)|^2$.  We prove the periodic analogue here in \cref{Bruno}, and the proofs are done the same way, except that Katz and Scharf's aperiodic sequences are identified with Laurent polynomials in $\C[z,z^{-1}]$, while our periodic sequences of length $\ell$ are identified with elements in the quotient ring $\quotring{\ell}$.}

Sarwate showed \cite[eqs.~(13),(38)]{Sarwate-1984-Mean} that, for a binary sequence $f$ of length $\ell$ selected at random, the expected value of $\ADF(f)$ is $1-1/\ell$, and if $f$ and $g$ are a randomly selected pair of binary sequences of length $\ell$, then the expected value of $\CDF(f,g)$ is $1$.
This means that if we select at random a binary codebook $F$ of length $\ell$, the expected value of $\CDF(F)$ will be $1+1/|F|-(1/|F|\ell)$.
We are interested in values of $\CDF(F)$ where $F$ runs through an infinite family $\calF$ of uniform-length, unimodular codebooks, with each codebook having the same number of sequences.
For certain such families $\calF$, we investigate the asymptotic value of $\CDF(F)$ as the length of the codebook tends to infinity.
This can give a good sense of the performance of the codebooks of moderate length taken from $\calF$; see the figures in \cite{Boothby-Katz} for many examples where asymptotic behavior gives a good prediction of the demerit factors for the majority of sequences of length less than $2000$.
If $\calF$ were made up of randomly selected uniform-length binary codebooks of $N$ sequences each, the expected value of $\CDF(F)$ would tend to $1+1/N$ as the length of the codebook approaches infinity.
We want families where $\CDF(F)$ tends to a limit much closer to the absolute minimum value of $1$ than to this typical value of $1+1/N$.

\subsection{Periodic sequences}

In this paper, we produce aperiodic sequences from periodic sequences, which we now define.
For $\ell \in \Z^+$, a {\it periodic sequence of length $\ell$} is a function $g\colon \Zell \to \C$, where we usually write $g_j$ instead of $g(j)$, and we write $g$ out as $(g_0,g_1,\ldots,g_{\ell-1})$.
If $k \in \Z$, we use the convention that $g_k$ means $g_{k+\ell\Z}$; here, $k+\ell\Z$ is the reduction of $k$ modulo $\ell$.
If $f$ is a periodic sequence, then $\len(f)$ denotes the length of $f$.

For a fixed $\ell$, the set of periodic sequences of length $\ell$ is well known to be a $\C$-inner product space under the usual componentwise addition, $\C$-scalar multiplication, and dot product.
Because of this, we may speak of pairs or sets of periodic sequences as being {\it orthogonal}, with the usual meaning.
A {\it balanced} periodic sequence $g=(g_0,\ldots,g_{\ell-1})$ is one whose terms sum to zero, or equivalently, which is orthogonal to the periodic sequence $(1,\ldots,1)$.

A {\it unimodular periodic sequence} is one whose terms are all complex numbers of magnitude $1$.
If $m \in \Z^+$, then an {\it $m$-ary periodic sequence} is one whose terms all lie in the $m$-ary alphabet.

A {\it periodic codebook} is a finite set of periodic sequences.
Any property (e.g., length, balance, unimodularity, $m$-arity) that can be predicated of a periodic sequence is also predicated of a periodic codebook if and only if every sequence in the codebook has that property.
So, for example, we can say that $F$ is a balanced, binary periodic codebook with $\len(F)=7$ if and only if every sequence in $F$ is a balanced, binary periodic sequence of length $7$.
A {\it uniform-length periodic codebook} is a periodic codebook in which all sequences have the same length (so that the codebook has a defined length).

If $g$ is a periodic sequence of length $\ell$ and $r \in \Z$, then {\it $g$ rotated by $r$}, written $g^{(r)}$, is the aperiodic sequence with $g^{(r)}_j=g_{r+j}$ for $j \in \{0,1,\ldots,\ell-1\}$ and $g^{(r)}_j=0$ for $j\not\in\{0,1,\ldots,\ell-1\}$, i.e.,
\[
g^{(r)}=(\ldots,0,0,g_r,g_{r+1},\ldots,g_{r+\ell-1},0,0,\ldots).
\]
A {\it rotation} of $g$ is a sequence $g^{(r)}$ for some $r \in \Z$, and $r$ is called the {\it advancement}.
If $G$ is a periodic codebook, then a {\it rotation} of $G$ is an aperiodic codebook of the form $\{g^{(r_g)}: g \in G\}$ where $r_g$ is some integer for each $g \in G$.
That is, each sequence in $G$ is replaced by one of its rotations.

\subsection{Cyclotomic sequences}

The periodic sequences of interest to us are defined using cyclotomic classes of finite fields.
If $p$ is a prime, then $\Fp$ is $\Z/p\Z$ (the finite field of order $p$) and $\Fpu$ is the unit group of $\Fp$.
For each prime $p$, we let $\alpha_p$ be a fixed primitive element (generator of $\Fpu$).
If $n \in \Z^+$ and $p$ is a prime with $p\equiv 1 \pmod{n}$, then $\Fpn=\{a^n: a \in \Fpu\}$, which is a subgroup of index $n$ in $\Fpu$.
The cosets of $\Fpn$ in $\Fpu$ are called {\it cyclotomic classes of index $n$}.
The {\it $k$th cyclotomic class of index $n$ in $\Fp$} is $\alpha_p^k \Fpn$, and the quotient group $\Fpu/\Fpn$ is $\{\alpha_p^0 \Fpn, \ldots,\alpha_p^{n-1} \Fpn\}$.

If $n \in \Z^+$, then a {\it cyclotomic pattern of index $n$} is a periodic sequence of length $n$.
If $d=(d_0,d_1,\ldots,d_{n-1})$ is a cyclotomic pattern of index $n$ and $p$ is a prime with $p\equiv 1\pmod{n}$, then the {\it periodic sequence of length $p$ (derived) from $d$} is the periodic sequence $f\colon \Fp \to \C$ with
\[
f_h = \begin{cases}
0 & \text{if $h=0$,} \\
d_k & \text{if $h \in \alpha_p^k\Fpn$.}
\end{cases}
\]
If $p$ is a prime with $p\equiv 1 \pmod{n}$, then an {\it aperiodic sequence (derived) from $d$ via $p$} is any rotation of the periodic sequence of length $p$ derived from $d$.
An aperiodic sequence derived from $d$ via $p$ is also called a {\it $p$-instance of $d$}, or just an {\it instance of $d$} if the particular prime need not be specified.

A {\it cyclotomic plan of index $n$} is a finite set of cyclotomic patterns of index $n$.
We say that cyclotomic plan is {\it unimodular} (resp., {\it $m$-ary}) to mean that every cyclotomic pattern in the plan is unimodular (resp., $m$-ary).
If $D$ is a cyclotomic plan of index $n$, and $p$ is a prime with $p\equiv 1 \pmod{n}$, then the {\it periodic codebook of length $p$ (derived) from $D$} is the set of $|D|$ periodic sequences of length $p$ that are derived from the cyclotomic patterns in $D$.
An {\it aperiodic codebook (derived) from $D$ via $p$} is a rotation of the periodic codebook of length $p$ derived from $D$.
An aperiodic codebook derived from $D$ via $p$ is also called a {\it $p$-instance of $D$}, or just an {\it instance of $D$} if the particular prime need not be specified.
Sequences from cyclotomic patterns are the same as sequences derived from linear combinations of multiplicative characters of finite fields (see \cref{Marlene} for a proof of this fact).
Small codebooks consisting of three sequences from the cyclotomic plan $\{(1,1,-1,-1),(1,-1,-1,1),(1,-1,1,-1)\}$ were studied by Boothby and Katz in \cite{Boothby-Katz}.

\subsection{Unimodularization}\label{Una}
If $d$ is a cyclotomic pattern and $f$ is a periodic sequence derived from $d$, then we have $f_0=0$.
All the other terms of $f$ will be unimodular (resp., in the $m$-ary alphabet for some $m \in \Z^+$) if and only if the pattern $d$ is unimodular (resp., $m$-ary).
For applications, we usually want all terms, including $f_0$, to be unimodular or $m$-ary, so we introduce the process of {\it unimodularization}.
A periodic sequence $f$ of length $\ell$ is said to be {\it unimodularizable} if $|f_j|=1$ for all nonzero $j\in \Zell$.
A {\it unimodularization} of this periodic sequence $f$ is any periodic sequence $\uni{f}$ with $\uni{f}_j=f_j$ for every nonzero $j\in\Zell$, and $|\uni{f}_0|=1$.
If $d$ is a unimodular cyclotomic pattern and $f$ is the periodic sequence of length $p$ derived from $d$, then any unimodularization $\uni{f}$ of $f$ is a {\it unimodularized periodic sequence of length $p$ (derived) from $d$}, and a rotation of $\uni{f}$ is a {\it unimodularized aperiodic sequence (derived) from $d$ via $p$}, also known as a {\it unimodularized $p$-instance of $d$} or just a {\it unimodularized instance of $d$} if the particular prime need not be specified.
Note that a unimodularized $p$-instance of $d$ is a unimodular sequence of length $p$.

A periodic codebook consisting entirely of unimodularizable sequences is said to be a {\it unimodularizable periodic codebook}.
If $F$ is such a codebook, then a {\it unimodularization} of $F$ is a periodic codebook obtained by replacing each sequence $f$ in $F$ with a unimodularization of $f$.
If $D$ is a unimodular cyclotomic plan, then a {\it unimodularized periodic codebook of length $p$ (derived) from $D$} is a unimodularization of the periodic codebook of length $p$ derived from $D$.
A {\it unimodularized aperiodic codebook (derived) from $D$ via $p$} is a rotation of a unimodularized periodic codebook of length $p$ derived from $D$.
A unimodularized aperiodic codebook derived from $D$ via $p$ is also known as a {\it unimodularized $p$-instance of $D$}, or just a {\it unimodularized instance of $D$} if the particular prime need not be specified.
Note that a unimodularized $p$-instance of $D$ is a unimodular codebook of length $p$.
Also note that any periodic or unimodularized periodic codebook derived from a unimodular cyclotomic plan $D$ must always have $|D|$ sequences, but it is possible for an aperiodic codebook or unimodularized aperiodic codebook derived from $D$ to have fewer than $|D|$ sequences.
For example, cyclotomic patterns $(1,0)$ and $(0,1)$ produce the length $p=3$ periodic sequences $f=(0,1,0)$ and $g=(0,0,1)$, respectively, and $f^{(1)}=g^{(2)}$.
However, the possibility of having fewer than $|D|$ sequences can befall $p$-instances or unimodularized $p$-instances of $D$ only for a finite set of small primes $p$ (see \cref{Nora}), so this phenomenon is not of importance when we consider asymptotic results as $p$ tends to infinity.

\subsection{Results}\label{Rudolph}

This paper has two main results.
The first, which we prove as \cref{Samuel}, is a criterion that tells us which unimodular cyclotomic plans produce families of codebooks with well regulated growth of $\GUC$.
\begin{theorem}\label{Penelope}
Let $D$ be a unimodular cyclotomic plan.
Let $\calA$ be a family of unimodularized instances of $D$ such that the $\{\len(A): A \in \calA\}$ is unbounded.
Then $\calA$ has well regulated growth of $\GUC$ if and only if $D$ is both balanced and orthogonal.
\end{theorem}
We also have a result about the mean square correlation for codebooks from unimodular cyclotomic plans.
Because \cref{Penelope} tells us that we should use orthogonal, balanced plans, if our plan is of index $n$, we cannot have more than $n-1$ patterns in it, so we make our codebooks as large as possible by insisting that our plan has exactly $n-1$ patterns.
A {\it Hadamard plan of index $n$} is an orthogonal, balanced, unimodular cyclotomic plan of index $n$ containing precisely $n-1$ cyclotomic patterns.
In this case, the patterns in our plan, along with the pattern $(1,1,\ldots,1)$, can be written as rows of an $n\times n$ matrix $M$ with $M M^*=n I$; this will be a Hadamard matrix (with one row being all ones) if our patterns are binary.
Since $n\times n$ Hadamard matrices are known to exist for $n$ equal to all multiples of $4$ up to and including $664$ \cite{Kharaghani-Tayfeh-Rezaie} (and since any Hadamard matrix can be converted into a Hadamard matrix with an all-ones row by negating some of the columns), this allows for plenty of binary codebook constructions.

To state our result on mean square correlation of codebooks derived from a Hadamard plan, we must specify more precisely how the rotation is being done when we create instances of our plan.
If $D$ is a cyclotomic plan, we review how one obtains a family of instances (or unimodularized instances) of $D$.
First we let $\{P_\iota\}_{\iota \in I}$ be a family of periodic codebooks (or a family of unimodularized periodic codebooks) from $D$.
For each $\iota \in I$, we let $p_\iota=\len(P_\iota)$, and for each $d \in D$ we let $f_{\iota,d}$ denote the periodic sequence (or unimodularized periodic sequence) of length $p_\iota$ from $d$ such that $P_\iota=\{f_{\iota,d}: d \in D\}$.
For each $\iota \in I$, we let $r_\iota \colon D \to \Z$ be a function, and let $A_\iota$ be the aperiodic codebook $\{f_{\iota,d}^{(r_\iota(d))}: d \in D\}$, which is a rotation of $P_\iota$.
So far this is a completely generic construction.
If $\{p_\iota: \iota \in I\}$ is unbounded and if there is some $\rho \in \R$ such that, for every $d \in D$, the quantity $r_{\iota}(d)/p_\iota$ tends to $\rho$ as $p_\iota$ tends to infinity, then the process we just described is a {\it coherently $\rho$-rotated construction of $\{A_\iota\}_{\iota \in I}$}.
A family of instances (or family unimodular instances) of a cyclotomic plan $D$ is said to be {\it coherently $\rho$-rotated} if there is a coherently $\rho$-rotated construction of it.
Since rotation is a cyclic process modulo the length of the underlying periodic sequence, any coherently $\rho$-rotated family will also be a coherently $(\rho+m)$-rotated family for each $m \in \Z$, and such a family can also be obtained by other constructions that are not coherently rotated.
A very simple way to obtain a coherently $\rho$-rotated family from $D$ is to rotate all the sequences in a periodic codebook (or unimodularized periodic codebook) of length $p$ from $D$ by $\floor{\rho p}$ (or, equally well $\ceil{\rho p}$) to generate a $p$-instance (or unimodularized $p$-instance) of $D$; doing this for an infinite number of distinct primes produces a coherently $\rho$-rotated family from $D$.
We define a periodic function $\Phi \colon \R \to \R$ with
\begin{equation}\label{Philip}
\begin{aligned}
\Phi(x+1) & =\Phi(x) & & \text{for every $x \in \R$,} \\
\Phi(x) & = 2\left(x-\frac{1}{2}\right)^2-\frac{1}{6} & & \text{for $0 \leq x \leq 1$,}
\end{aligned}
\end{equation}
and now we can state our second main result, which we obtain as part of \cref{Julia}.
\begin{theorem}\label{Mary}
Let $n \in \Z^+$, let $D$ be a Hadamard plan of index $n$, and let $\rho \in \R$.
Let $\{A_\iota\}_{\iota \in A}$ be coherently $\rho$-rotated family of unimodularized instances of $D$.
Then $\CDF(A_\iota)$ tends to
\[
1+ \frac{1}{3(n-1)} + \frac{1}{n-1} \Phi\left(2 \rho\right).
\]
as $\len(A_\iota)$ tends to infinity.
For a fixed $n$, the right hand side achieves its global minimum value of $1+1/(6(n-1))$ if and only if $\rho \in \{(2n+1)/4: n \in \Z\}$.
\end{theorem}

\subsection{Examples}

To illustrate our results, we directly calculated correlation data for codebooks derived from Hadamard plans.
For $n \in \N$, the {\it $n$th Walsh Hadamard matrix}, written $H_n$, is the $2^n\times 2^n$ matrix with entries in $\{-1,1\}$ defined recursively by $H_0=\begin{bmatrix} 1 \end{bmatrix}$ and
\[
H_{n+1} = \begin{bmatrix} H_n & H_n \\ H_n & -H_n \end{bmatrix}.
\]
For $n \in \N$, the {\it $n$th Walsh Hadamard plan}, written $D_n$, is the balanced, orthogonal, binary plan of index $2^n$ whose $2^n-1$ constituent patterns are obtained by considering the rows of $H_n$ as periodic sequences of order $2^n$, and discarding the first row (which is $(1,1,\ldots,1)$).
We note that this construction makes $D_1=\{(1,-1)\}$ which produces codebooks with one sequence apiece, namely the Legendre sequence, and so our results recapitulate H\o holdt and Jensen's results \cite{Hoholdt-Jensen} on the autocorrelation of Legendre sequences in this case.
The Walsh Hadamard construction also makes $D_2=\{(1,-1,1,-1),(1,1,-1,-1),(1,-1,-1,1)\}$, which produces codebooks with three sequences apiece: the Legendre sequence and two other sequences that can be derived from quartic characters; these were studied in \cite{Boothby-Katz}.

In this section, we consider unimodular $p$-instances of $D_n$ for various $p$ and $n$, where we always unimodularize by replacing the $0$ entry of each periodic sequence with a $1$.
So our codebooks formed from $D_n$ are always binary and of length $p$.
When forming the underlying periodic sequences, the primitive element $\alpha_p$ of $\Fp$ that we use is $j+p\Z$, where $j$ is the least positive integer such that $j+p\Z$ is primitive in $\Fp$.
When forming our $p$-instances, we always rotate every sequence in the codebook by the same advancement, so one can say that these codebooks are {\it uniformly rotated}.
As such, we list one advancement for the whole codebook.

When judging the quality of a codebook $F$ with respect to peak measures, we use $\GUC(F)/\sqrt{\SDC(F)}$, since we want this number to be small and not grow too rapidly as sequence length increases.
We use $|F|(\CDF(F)-1)$ to evaluate the crosscorrelation demerit factor: in this way codebooks consisting of $N$ randomly selected sequences of length $\ell$ would be expected to have $N(\CDF(F)-1)=1-1/\ell$, which is close to $1$ when $\ell$ is large.
This allows us to compare codebooks with different numbers of sequences on a similar basis.
We call $|F|(\CDF(F)-1)$ the {\it adjusted demerit factor} of the codebook.
Throughout this section, including all tables of this paper, any number written with a decimal point is to be understood as a decimal approximation of a quantity, but not necessarily an exact value; we omit the trailing $\ldots$ that we would normally write for approximations in order to avoid clutter, especially in the tables.

To check the rotation-dependence in the formula of \cref{Mary}, we form unimodularized $1009$-instances of $D_3$ with all possible advancements from $r=0$ to $r=1008$.
\cref{Mary} predicts that if $r/p$ tends to $\rho$ as $p$ tends to infinity in an infinite family of unimodularized $p$-instances of $D_3$, then the adjusted demerit factor should tend to $(1/3)+\Phi(2\rho)$.
In \cref{Claire}, we plot (as dots) the adjusted demerit factors of our codebooks with $p=1009$ as a function of $x=r/p$ for $0 \leq r < p$, and also display the curve of the function $f(x)=(1/3)+\Phi(2 x)$ for comparison; this shows that the behavior of the codebooks is close to the asymptotic prediction even at a modest length.
The adjusted demerit factor ranges from $0.166092$ at advancements $r=278$ and $732$ (so $r/p=0.275520$ and $0.725471$) to $0.722886$ at advancements $r=501$ and $509$ (so $r/p=0.496531$ and $0.504460$).
\begin{figure}[ht!]
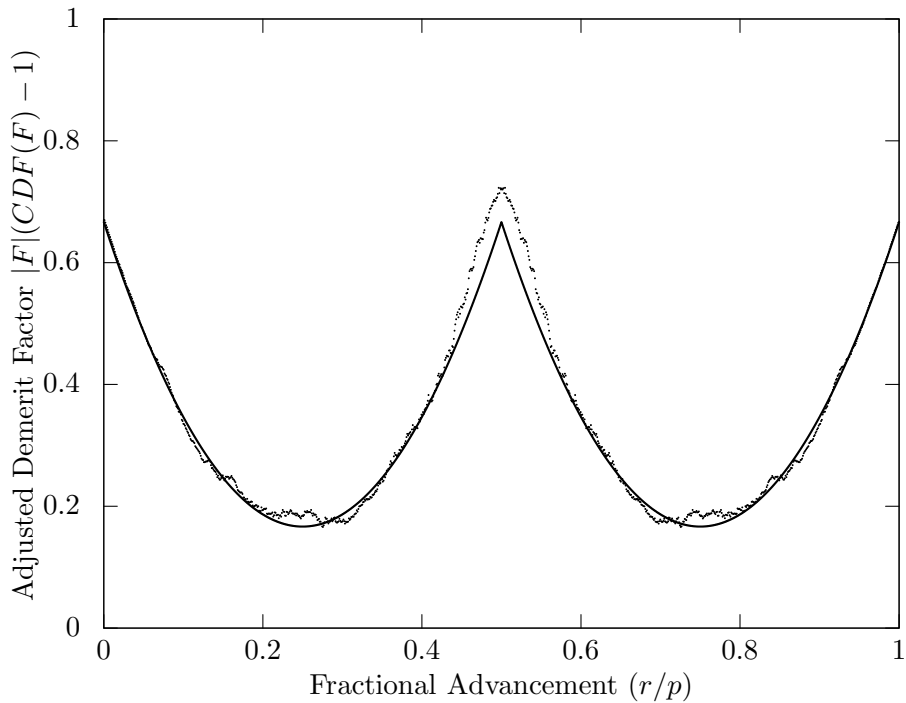

\caption{Dependence of $\CDF(F)$ on advancement $r$ for unimodularized instances of length $1009$ from $D_3$}\label{Claire}
\begin{center}

\end{center}
\end{figure}

To see how rapidly the demerit factor approaches its limit, as predicted by \cref{Mary}, we make unimodularized $p$-instances of $D_3$, with each instance uniformly rotated by $(p-1)/4$ in order to make $r/p$ close to the optimizing limiting value of $1/4$.
\cref{Mary} predicts that the adjusted demerit factor of such codebooks should approach $1/6$ as $p$ tends to infinity.
And \cref{Penelope} predicts that these codebooks should have well regulated growth of $\GUC$.
Since $D_3$ is of index $8$, these constructions are possible for every prime $p$ with $p\equiv 1 \pmod{8}$, and produce codebooks with $7$ sequences each.
We tabulate correlation data for these constructions for all primes $p$ with $p\equiv 1 \pmod{8}$ and $p \leq 1249$ in \cref{Heinrich} (for peak measures) and \cref{Otto} (for demerit factors).
The columns labeled ``sequence PSLs'' in \cref{Heinrich} (resp., ``sequence ADFs'' in \cref{Otto}) give the average, minimum, and maximum values of $\PSL(f)$ (resp., $\ADF(f)$) as $f$ runs through the individual sequences in our codebook.
The columns labeled ``pairwise PCCs'' in \cref{Heinrich} (resp., ``pairwise CDFs'' in \cref{Otto}) give the average, minimum, and maximum values of $\PCC(f,g)$ (resp., $\CDF(f,g)$) as $(f,g)$ runs through the pairs of distinct sequences in our codebook.
Note that $\GUC(F)/\sqrt{\SDC(F)}$ is never much higher than $3$, and the adjusted demerit factors are getting close to the limit of $1/6$ predicted by \cref{Mary}.
To see the limiting behavior even more strongly, we also checked some larger primes: for each $k$ with $4 \leq k \leq 25$, we let $p_k$ be the smallest prime with $p_k \geq 2^k+1$ and $p_k\equiv 1 \pmod{8}$.
The data for these primes is given on \cref{Irmgard} (for peak measures) and \cref{Hildegard} (for demerit factors).
\begin{table}\caption{Peak measures for unimodularized $p$-instances of $D_3$, rotated by $(p-1)/4$}\label{Heinrich}
\begin{center}
{\footnotesize 
\begin{tabular}{r|r|rrr|rrr}
    &                      & \multicolumn{3}{c|}{sequence $\PSL$s} & \multicolumn{3}{c}{pairwise $\PCC$s} \\
$p$ & \multirow{-2}{*}{$\frac{\GUC(F)}{\sqrt{\SDC(F)}}$} & avg & min & max            & avg & min & max \\
\hline
$17$ & $2.910428$ & $6.0000$ & $4$ & $8$ & $7.1429$ & $4$ & $12$ \\
$41$ & $2.342606$ & $10.8571$ & $7$ & $13$ & $11.3333$ & $9$ & $15$ \\
$73$ & $2.457864$ & $13.7143$ & $10$ & $19$ & $15.6667$ & $11$ & $21$ \\
$89$ & $2.543995$ & $11.4286$ & $8$ & $15$ & $18.6190$ & $14$ & $24$ \\
$97$ & $2.741435$ & $13.7143$ & $8$ & $21$ & $19.9048$ & $15$ & $27$ \\
$113$ & $2.445874$ & $16.8571$ & $11$ & $23$ & $20.8571$ & $12$ & $26$ \\
$137$ & $2.306766$ & $18.0000$ & $10$ & $25$ & $22.9048$ & $16$ & $27$ \\
$193$ & $2.807281$ & $19.7143$ & $14$ & $28$ & $29.1905$ & $18$ & $39$ \\
$233$ & $2.751511$ & $21.4286$ & $13$ & $29$ & $30.8095$ & $20$ & $42$ \\
$241$ & $2.769873$ & $25.4286$ & $12$ & $41$ & $31.7143$ & $23$ & $43$ \\
$257$ & $3.056536$ & $26.7143$ & $17$ & $35$ & $34.3333$ & $21$ & $49$ \\
$281$ & $2.684475$ & $25.8571$ & $18$ & $33$ & $36.3333$ & $26$ & $45$ \\
$313$ & $2.543550$ & $30.4286$ & $15$ & $45$ & $36.1429$ & $28$ & $44$ \\
$337$ & $2.887094$ & $33.8571$ & $20$ & $45$ & $38.5714$ & $26$ & $53$ \\
$353$ & $2.820905$ & $32.8571$ & $16$ & $53$ & $38.7143$ & $27$ & $47$ \\
$401$ & $2.546818$ & $35.1429$ & $19$ & $49$ & $39.6667$ & $22$ & $51$ \\
$409$ & $3.016256$ & $33.4286$ & $18$ & $47$ & $42.0000$ & $25$ & $61$ \\
$433$ & $2.979530$ & $34.1429$ & $20$ & $48$ & $44.9524$ & $27$ & $62$ \\
$449$ & $2.925961$ & $32.8571$ & $18$ & $42$ & $48.3810$ & $33$ & $62$ \\
$457$ & $2.619570$ & $31.4286$ & $21$ & $44$ & $45.8571$ & $35$ & $56$ \\
$521$ & $2.628648$ & $37.2857$ & $19$ & $60$ & $46.6190$ & $34$ & $54$ \\
$569$ & $2.892630$ & $35.7143$ & $25$ & $43$ & $51.5714$ & $39$ & $69$ \\
$577$ & $2.747616$ & $40.0000$ & $22$ & $55$ & $52.8095$ & $31$ & $66$ \\
$593$ & $2.628165$ & $43.1429$ & $22$ & $60$ & $53.4762$ & $39$ & $64$ \\
$601$ & $2.896150$ & $42.2857$ & $23$ & $60$ & $55.2857$ & $41$ & $71$ \\
$617$ & $2.737577$ & $38.5714$ & $23$ & $51$ & $54.0000$ & $44$ & $68$ \\
$641$ & $2.804331$ & $31.7143$ & $22$ & $39$ & $57.3810$ & $45$ & $71$ \\
$673$ & $2.659755$ & $48.4286$ & $22$ & $66$ & $58.1429$ & $47$ & $69$ \\
$761$ & $2.827498$ & $38.7143$ & $22$ & $47$ & $60.8095$ & $48$ & $78$ \\
$769$ & $2.596386$ & $51.7143$ & $25$ & $72$ & $60.8095$ & $42$ & $72$ \\
$809$ & $2.601701$ & $53.7143$ & $29$ & $74$ & $58.6667$ & $38$ & $72$ \\
$857$ & $2.732748$ & $47.5714$ & $28$ & $62$ & $62.7143$ & $44$ & $80$ \\
$881$ & $2.594196$ & $51.0000$ & $26$ & $66$ & $65.7143$ & $49$ & $77$ \\
$929$ & $2.887186$ & $43.4286$ & $31$ & $49$ & $71.7143$ & $58$ & $88$ \\
$937$ & $2.776829$ & $55.4286$ & $28$ & $72$ & $66.5238$ & $48$ & $85$ \\
$953$ & $2.785813$ & $48.8571$ & $31$ & $68$ & $68.3333$ & $46$ & $86$ \\
$977$ & $2.687398$ & $54.1429$ & $27$ & $78$ & $72.5238$ & $58$ & $84$ \\
$1009$ & $2.896291$ & $52.5714$ & $28$ & $71$ & $73.9048$ & $52$ & $92$ \\
$1033$ & $2.893562$ & $48.5714$ & $30$ & $62$ & $73.6190$ & $47$ & $93$ \\
$1049$ & $2.686158$ & $57.0000$ & $29$ & $76$ & $70.1905$ & $51$ & $87$ \\
$1097$ & $2.475771$ & $39.0000$ & $26$ & $49$ & $70.8571$ & $64$ & $82$ \\
$1129$ & $3.065422$ & $58.4286$ & $33$ & $83$ & $76.1429$ & $46$ & $103$ \\
$1153$ & $2.827200$ & $62.7143$ & $30$ & $87$ & $79.2857$ & $53$ & $96$ \\
$1193$ & $2.692543$ & $65.0000$ & $31$ & $93$ & $75.7619$ & $58$ & $89$ \\
$1201$ & $2.885549$ & $67.5714$ & $29$ & $98$ & $77.3810$ & $63$ & $100$ \\
$1217$ & $2.723192$ & $56.5714$ & $34$ & $77$ & $76.0000$ & $52$ & $95$ \\
$1249$ & $2.999333$ & $54.8571$ & $32$ & $68$ & $84.2381$ & $59$ & $106$
\end{tabular}}
\end{center}
\end{table}

\begin{table}\caption{Demerit factors for unimodularized $p$-instances of $D_3$, rotated by $(p-1)/4$}\label{Otto}
\begin{center}
{\footnotesize 
\begin{tabular}{r|r|rrr|rrr}
    &  adj.~demerit factor & \multicolumn{3}{c|}{sequence $\ADF$s} & \multicolumn{3}{c}{pairwise $\CDF$s} \\
$p$ & $|F|(\CDF(F)-1)$ & avg & min & max            & avg & min & max \\
\hline
$17$ & $0.419179$ & $0.8463$ & $0.3599$ & $1.4948$ & $0.9288$ & $0.3772$ & $1.5536$ \\
$41$ & $0.219597$ & $0.7268$ & $0.2094$ & $1.2136$ & $0.9155$ & $0.5526$ & $1.4521$ \\
$73$ & $0.275152$ & $0.8160$ & $0.2297$ & $1.4817$ & $0.9099$ & $0.4220$ & $1.6140$ \\
$89$ & $0.263892$ & $0.4425$ & $0.1828$ & $0.6575$ & $0.9702$ & $0.5849$ & $1.2626$ \\
$97$ & $0.257504$ & $0.5267$ & $0.2236$ & $0.8826$ & $0.9551$ & $0.5821$ & $1.2912$ \\
$113$ & $0.180616$ & $0.6901$ & $0.2105$ & $1.1490$ & $0.9151$ & $0.2883$ & $1.3521$ \\
$137$ & $0.184073$ & $0.5827$ & $0.1782$ & $0.9752$ & $0.9336$ & $0.5075$ & $1.2244$ \\
$193$ & $0.207530$ & $0.5494$ & $0.1862$ & $0.8196$ & $0.9430$ & $0.4887$ & $1.3327$ \\
$233$ & $0.188431$ & $0.6051$ & $0.1843$ & $0.9114$ & $0.9306$ & $0.3320$ & $1.6031$ \\
$241$ & $0.172567$ & $0.6373$ & $0.1814$ & $1.0650$ & $0.9225$ & $0.3616$ & $1.3896$ \\
$257$ & $0.130258$ & $0.5929$ & $0.1860$ & $0.8142$ & $0.9229$ & $0.3660$ & $1.2289$ \\
$281$ & $0.199288$ & $0.7084$ & $0.1788$ & $1.1954$ & $0.9152$ & $0.3937$ & $1.5160$ \\
$313$ & $0.219136$ & $0.6808$ & $0.1769$ & $1.2906$ & $0.9231$ & $0.5033$ & $1.3680$ \\
$337$ & $0.195768$ & $0.7807$ & $0.1823$ & $1.3801$ & $0.9025$ & $0.4228$ & $1.4782$ \\
$353$ & $0.226904$ & $0.6448$ & $0.1755$ & $1.0826$ & $0.9303$ & $0.3685$ & $1.2834$ \\
$401$ & $0.163595$ & $0.8365$ & $0.1735$ & $1.4757$ & $0.8879$ & $0.3044$ & $1.5537$ \\
$409$ & $0.160121$ & $0.7280$ & $0.1756$ & $1.2043$ & $0.9054$ & $0.3382$ & $1.4370$ \\
$433$ & $0.191706$ & $0.6190$ & $0.1749$ & $0.9584$ & $0.9288$ & $0.3271$ & $1.4667$ \\
$449$ & $0.157573$ & $0.6041$ & $0.1731$ & $0.7817$ & $0.9256$ & $0.4306$ & $1.4336$ \\
$457$ & $0.185041$ & $0.4240$ & $0.1698$ & $0.6865$ & $0.9602$ & $0.6111$ & $1.1179$ \\
$521$ & $0.177503$ & $0.7846$ & $0.1704$ & $1.3795$ & $0.8988$ & $0.4395$ & $1.4551$ \\
$569$ & $0.150238$ & $0.4650$ & $0.1746$ & $0.6594$ & $0.9475$ & $0.5153$ & $1.2512$ \\
$577$ & $0.165774$ & $0.6871$ & $0.1753$ & $1.0842$ & $0.9131$ & $0.3376$ & $1.3795$ \\
$593$ & $0.184112$ & $0.7016$ & $0.1724$ & $1.2221$ & $0.9138$ & $0.5492$ & $1.2136$ \\
$601$ & $0.189947$ & $0.6105$ & $0.1695$ & $0.8292$ & $0.9299$ & $0.3693$ & $1.3075$ \\
$617$ & $0.157510$ & $0.5380$ & $0.1722$ & $0.9240$ & $0.9366$ & $0.6464$ & $1.3116$ \\
$641$ & $0.174899$ & $0.2984$ & $0.1701$ & $0.4476$ & $0.9794$ & $0.8045$ & $1.0532$ \\
$673$ & $0.160722$ & $0.7072$ & $0.1720$ & $1.2373$ & $0.9089$ & $0.5742$ & $1.3069$ \\
$761$ & $0.161539$ & $0.3364$ & $0.1683$ & $0.5067$ & $0.9709$ & $0.6691$ & $1.0841$ \\
$769$ & $0.187306$ & $0.7163$ & $0.1707$ & $1.2557$ & $0.9118$ & $0.5015$ & $1.2510$ \\
$809$ & $0.195315$ & $0.8362$ & $0.1700$ & $1.4814$ & $0.8932$ & $0.3494$ & $1.5395$ \\
$857$ & $0.186538$ & $0.5818$ & $0.1710$ & $0.9220$ & $0.9341$ & $0.3570$ & $1.5243$ \\
$881$ & $0.218891$ & $0.6226$ & $0.1719$ & $1.0964$ & $0.9327$ & $0.4998$ & $1.2975$ \\
$929$ & $0.173817$ & $0.3822$ & $0.1688$ & $0.5446$ & $0.9653$ & $0.7085$ & $1.1580$ \\
$937$ & $0.184648$ & $0.7430$ & $0.1705$ & $1.2987$ & $0.9069$ & $0.3642$ & $1.3619$ \\
$953$ & $0.189571$ & $0.6129$ & $0.1696$ & $1.0190$ & $0.9294$ & $0.4720$ & $1.3698$ \\
$977$ & $0.174837$ & $0.6329$ & $0.1697$ & $1.0400$ & $0.9237$ & $0.4764$ & $1.1013$ \\
$1009$ & $0.184897$ & $0.6353$ & $0.1702$ & $1.0643$ & $0.9249$ & $0.4811$ & $1.3665$ \\
$1033$ & $0.174843$ & $0.5481$ & $0.1697$ & $0.8355$ & $0.9378$ & $0.3487$ & $1.2538$ \\
$1049$ & $0.172194$ & $0.8012$ & $0.1689$ & $1.4070$ & $0.8952$ & $0.3598$ & $1.5049$ \\
$1097$ & $0.186012$ & $0.2412$ & $0.1685$ & $0.3377$ & $0.9908$ & $0.8470$ & $1.0501$ \\
$1129$ & $0.161952$ & $0.6390$ & $0.1694$ & $1.0431$ & $0.9205$ & $0.3123$ & $1.3667$ \\
$1153$ & $0.170944$ & $0.6531$ & $0.1699$ & $1.1206$ & $0.9196$ & $0.4699$ & $1.1351$ \\
$1193$ & $0.189239$ & $0.7921$ & $0.1684$ & $1.3729$ & $0.8995$ & $0.3624$ & $1.4283$ \\
$1201$ & $0.176133$ & $0.7488$ & $0.1704$ & $1.3415$ & $0.9046$ & $0.4762$ & $1.3491$ \\
$1217$ & $0.177954$ & $0.5958$ & $0.1684$ & $0.9243$ & $0.9304$ & $0.3242$ & $1.5475$ \\
$1249$ & $0.142505$ & $0.5114$ & $0.1694$ & $0.7237$ & $0.9385$ & $0.4527$ & $1.2567$
\end{tabular}}
\end{center}
\end{table}

\begin{table}\caption{Peak measures for unimodularized $p$-instances of $D_3$, rotated by $(p-1)/4$}\label{Irmgard}
\begin{center}
{\footnotesize 
\begin{tabular}{r|r|rrr|rrr}
    &                      & \multicolumn{3}{c|}{sequence $\PSL$s} & \multicolumn{3}{c}{pairwise $\PCC$s} \\
$p$ &  \multirow{-2}{*}{$\frac{\GUC(F)}{\sqrt{\SDC(F)}}$} & avg & min & max            & avg & min & max \\
\hline
$17$ & $2.910428$ & $6.0000$ & $4$ & $8$ & $7.1429$ & $4$ & $12$ \\
$41$ & $2.342606$ & $10.8571$ & $7$ & $13$ & $11.3333$ & $9$ & $15$ \\
$73$ & $2.457864$ & $13.7143$ & $10$ & $19$ & $15.6667$ & $11$ & $21$ \\
$137$ & $2.306766$ & $18.0000$ & $10$ & $25$ & $22.9048$ & $16$ & $27$ \\
$257$ & $3.056536$ & $26.7143$ & $17$ & $35$ & $34.3333$ & $21$ & $49$ \\
$521$ & $2.628648$ & $37.2857$ & $19$ & $60$ & $46.6190$ & $34$ & $54$ \\
$1033$ & $2.893562$ & $48.5714$ & $30$ & $62$ & $73.6190$ & $47$ & $93$ \\
$2081$ & $2.652463$ & $87.1429$ & $45$ & $116$ & $105.0000$ & $71$ & $121$ \\
$4129$ & $3.096925$ & $108.2857$ & $62$ & $153$ & $166.7143$ & $111$ & $199$ \\
$8209$ & $3.355278$ & $173.0000$ & $106$ & $237$ & $231.7143$ & $170$ & $304$ \\
$16417$ & $3.051616$ & $185.5714$ & $121$ & $215$ & $321.0476$ & $278$ & $391$ \\
$32801$ & $3.390196$ & $321.1429$ & $196$ & $400$ & $499.1905$ & $413$ & $614$ \\
$65537$ & $3.222632$ & $505.8571$ & $300$ & $577$ & $711.7143$ & $432$ & $825$ \\
$131113$ & $3.413466$ & $747.0000$ & $374$ & $1014$ & $980.3333$ & $817$ & $1236$ \\
$262153$ & $3.451113$ & $1017.2857$ & $608$ & $1282$ & $1423.1905$ & $926$ & $1767$ \\
$524353$ & $3.554649$ & $1619.4286$ & $827$ & $2144$ & $2111.6190$ & $1574$ & $2574$ \\
$1048601$ & $3.437459$ & $2288.2857$ & $1169$ & $3081$ & $2770.1905$ & $1827$ & $3520$ \\
$2097169$ & $3.429868$ & $3280.8571$ & $1674$ & $4154$ & $4294.4286$ & $3354$ & $4967$ \\
$4194329$ & $3.437978$ & $4690.0000$ & $2465$ & $6251$ & $5868.7143$ & $3959$ & $7041$ \\
$8388617$ & $3.323193$ & $6902.0000$ & $3406$ & $9050$ & $8574.0952$ & $6714$ & $9625$ \\
$16777289$ & $3.362786$ & $9121.8571$ & $5053$ & $12069$ & $12599.9048$ & $10674$ & $13774$ \\
$33554473$ & $3.442828$ & $13686.1429$ & $7265$ & $17680$ & $17759.1905$ & $14392$ & $19943$
\end{tabular}}
\end{center}
\end{table}

\begin{table}\caption{Demerit factors for unimodularized $p$-instances of $D_3$, rotated by $(p-1)/4$}\label{Hildegard}
\begin{center}
{\footnotesize 
\begin{tabular}{r|r|rrr|rrr}
    &  adj.~demerit factor & \multicolumn{3}{c|}{sequence $\ADF$s} & \multicolumn{3}{c}{pairwise $\CDF$s} \\
$p$ & $|F|(\CDF(F)-1)$ & avg & min & max            & avg & min & max \\
\hline
$17$ & $0.419179$ & $0.8463$ & $0.3599$ & $1.4948$ & $0.9288$ & $0.3772$ & $1.5536$ \\
$41$ & $0.219597$ & $0.7268$ & $0.2094$ & $1.2136$ & $0.9155$ & $0.5526$ & $1.4521$ \\
$73$ & $0.275152$ & $0.8160$ & $0.2297$ & $1.4817$ & $0.9099$ & $0.4220$ & $1.6140$ \\
$137$ & $0.184073$ & $0.5827$ & $0.1782$ & $0.9752$ & $0.9336$ & $0.5075$ & $1.2244$ \\
$257$ & $0.130258$ & $0.5929$ & $0.1860$ & $0.8142$ & $0.9229$ & $0.3660$ & $1.2289$ \\
$521$ & $0.177503$ & $0.7846$ & $0.1704$ & $1.3795$ & $0.8988$ & $0.4395$ & $1.4551$ \\
$1033$ & $0.174843$ & $0.5481$ & $0.1697$ & $0.8355$ & $0.9378$ & $0.3487$ & $1.2538$ \\
$2081$ & $0.154620$ & $0.6833$ & $0.1676$ & $1.1700$ & $0.9119$ & $0.4365$ & $1.1985$ \\
$4129$ & $0.176674$ & $0.5046$ & $0.1677$ & $0.7656$ & $0.9453$ & $0.4171$ & $1.1867$ \\
$8209$ & $0.172492$ & $0.6040$ & $0.1669$ & $0.7902$ & $0.9281$ & $0.4013$ & $1.5079$ \\
$16417$ & $0.154870$ & $0.2453$ & $0.1669$ & $0.3326$ & $0.9849$ & $0.8348$ & $1.0340$ \\
$32801$ & $0.166212$ & $0.4624$ & $0.1667$ & $0.7765$ & $0.9506$ & $0.6926$ & $1.1975$ \\
$65537$ & $0.164376$ & $0.5500$ & $0.1667$ & $0.8319$ & $0.9357$ & $0.3356$ & $1.3279$ \\
$131113$ & $0.164924$ & $0.6295$ & $0.1667$ & $1.0196$ & $0.9226$ & $0.3418$ & $1.3294$ \\
$262153$ & $0.166653$ & $0.6140$ & $0.1667$ & $0.8433$ & $0.9254$ & $0.3343$ & $1.2304$ \\
$524353$ & $0.165693$ & $0.7178$ & $0.1667$ & $1.1989$ & $0.9080$ & $0.3437$ & $1.2909$ \\
$1048601$ & $0.166435$ & $0.8224$ & $0.1667$ & $1.4442$ & $0.8907$ & $0.3330$ & $1.4636$ \\
$2097169$ & $0.166490$ & $0.6103$ & $0.1667$ & $0.9549$ & $0.9260$ & $0.3344$ & $1.4341$ \\
$4194329$ & $0.166757$ & $0.7813$ & $0.1667$ & $1.3353$ & $0.8976$ & $0.3334$ & $1.4272$ \\
$8388617$ & $0.166887$ & $0.7616$ & $0.1667$ & $1.3726$ & $0.9009$ & $0.4612$ & $1.3724$ \\
$16777289$ & $0.166861$ & $0.5152$ & $0.1667$ & $0.8844$ & $0.9419$ & $0.6332$ & $1.2216$ \\
$33554473$ & $0.166610$ & $0.7123$ & $0.1667$ & $1.2686$ & $0.9091$ & $0.5099$ & $1.3660$
\end{tabular}}
\end{center}
\end{table}

Our constructions give codebooks of sequences with exceptionally low correlation, but the codebooks have few sequences compared to the length of the sequences.
We propose that a codebook of this type might find use in ranging applications.
For example, the Global Positioning System (GPS) uses a codebook of $36$ Gold sequences of length $1023$ for its coarse/acquisition code \cite{Flores}, so it uses only $36$ out of $1025$ possible Gold sequences that one could form from the two underlying m-sequences used in its construction.
We made a codebook of $36$ sequences of comparable length from a unimodularized instance of $D_6$.
Because the length needs to be a prime $p$ with $p\equiv 1 \pmod{64}$, we chose length $p=1153$, uniformly rotated by $(p-1)/4=288$, and then selected $36$ out of the $63$ sequences in the codebook, derived from the patterns obtained from rows $1$, $2$, $3$, $4$, $6$, $8$, $9$, $12$, $13$, $16$, $17$, $18$, $19$, $21$, $23$, $24$, $25$, $27$, $29$, $31$, $32$, $34$, $35$, $36$, $42$, $43$, $44$, $46$, $48$, $50$, $51$, $52$, $58$, $59$, $60$, and $62$
of the Walsh Hadamard matrix $H_6$ (where the top row of $H_6$, which is not used in the Walsh Hadamard plan $D_6$, is numbered as row $0$).
We compare the correlation data of the Global Positioning System's codebook (labeled GPS) and our codebook (labeled WH).
We also combine both codebooks into a single codebook (labeled GPS+WH), which has $72$ sequence (some of length $1153$ and some of length $1023$), and evaluate the correlation measures for this codebook.
All this correlation data is presented on \cref{Felipe} (for peak measures) and \cref{Esteban} (for demerit factors), where the entry GPS/WH tracks the crosscorrelations for pairs of sequences where one sequence is from the GPS codebook and the other is from the WH codebook.
We see that our WH codebook has considerably lower adjusted demerit factor than the GPS codebook, with autocorrelation demerit factors being much lower on average, which is good for ranging.
Overall, the WH codebook demonstrates a higher variability of both peak measures and demerit factors.
The peak measures in the combined book GPS+WH are somewhat higher than in the individual books, but the adjusted demerit factor of GPS+WH is still much superior to that of GPS alone, so the new codebook WH might be used alongside the current GPS codebook without too much interference.

\begin{table}\caption{Peak measures for the GPS codebook and the WH codebook}\label{Felipe}
\begin{center}
{\footnotesize 
\begin{tabular}{r|r|rrr|rrr}
    &                      & \multicolumn{3}{c|}{sequence $\PSL$s} & \multicolumn{3}{c}{pairwise $\PCC$s} \\
$p$ &  \multirow{-2}{*}{$\frac{\GUC(F)}{\sqrt{\SDC(F)}}$} & avg & min & max            & avg & min & max \\
\hline
GPS    & $3.376649$ & $81.8333$ & $73$ & $98$ & $85.0556$ & $74$ & $108$ \\
WH     & $3.681250$ & $74.0833$ & $30$ & $108$ & $89.3079$ & $64$ & $125$ \\
GPS/WH & & & & & $102.0069$ & $78$ & $161$ \\
GPS+WH & $5.033708$ & $77.9583$ & $30$ & $108$ & $94.6987$ & $64$ & $161$
\end{tabular}}
\end{center}
\end{table}

\begin{table}\caption{Demerit factors for the GPS codebook and the WH codebook}\label{Esteban}
\begin{center}
{\footnotesize 
\begin{tabular}{r|r|rrr|rrr}
    &  adj.~demerit factor & \multicolumn{3}{c|}{sequence $\ADF$s} & \multicolumn{3}{c}{pairwise $\CDF$s} \\
$p$ & $|F|(\CDF(F)-1)$ & avg & min & max            & avg & min & max \\
\hline
GPS    & $0.947520$ & $1.0091$ & $0.8393$ & $1.2192$ & $0.9982$ & $0.8962$ & $1.1742$ \\
WH     & $0.358546$ & $0.6998$ & $0.1699$ & $0.9686$ & $0.9902$ & $0.4840$ & $1.4887$ \\
GPS/WH & & & & & $0.9996$ & $0.9064$ & $1.1209$ \\
GPS+WH & $0.640059$ & $0.8545$ & $0.1699$ & $1.2192$ & $0.9970$ & $0.4840$ & $1.4887$
\end{tabular}}
\end{center}
\end{table}

To compute the correlation spectra for the various examples in this section, we used Fourier transform techniques that carry out convolutions using floating point arithmetic to approximate operations in $\C$.
Since the sequences are binary, the correlation values should be integral, and we found that all terms of our computed correlation spectra are very close to integral, with the largest discrepancy being less than $8 \times 10^{-6}$.

\subsection{Organization of this paper}

In \cref{Peter} on preliminaries, we provide definitions, notations, and concepts beyond those given in this introduction, demonstrate that sequences from cyclotomic patterns are the same as sequences given by linear combinations of multiplicative finite field characters, and show that unimodularization does not affect the sort of asymptotic behaviors we are considering in Theorems \ref{Penelope} and \ref{Mary}.
\cref{Bernard} is devoted to the proof of \cref{Penelope}.
\cref{Genevieve} is devoted to the proof of \cref{Mary}.

\section{Preliminaries}\label{Peter}

In this section, we add to the the definitions and notations from the Introduction, which all remain in force.
In \cref{Norman} we discuss norms and inner products for aperiodic and periodic sequences.
In \cref{Marlene} we use Fourier analysis to show that sequences from cyclotomic patterns are the same as sequences from linear combinations of multiplicative finite field characters.
In \cref{Estelle} we show the effect of unimodularization of sequences on the peak and mean square measures of correlation, and conclude that unimodularization does not affect the asymptotic behavior of these measures in the limits that we are considering.

\subsection{Norms and inner products}\label{Norman}

Let $q$ be a real number with $q \geq 1$.
The $l^q$ norm of an aperiodic sequence $f$ is $\norm{q}{f}=\left(\sum_{j \in \Z} |f_j|^q\right)^{1/q}$, and the $l^\infty$ norm of $f$ is $\norminf{f}=\max_{j \in \Z} |f_j|$; the finite support of $f$ makes both of these quantities defined.
If $g$ is a periodic sequence of length $\ell$, then the $l^q$ norm of $g$ is $\norm{q}{g}=\left(\sum_{j \in \Zell} |g_j|^q\right)^{1/q}$, and the $l^\infty$ norm of $g$ is $\norminf{g}=\max_{j \in \Zell} |g_j|$.
We should caution that the $l^q$ norms defined here are not the same as the $L^q$ norms on the complex unit circle that are used in studies of correlation, e.g., in \cite{Katz,Boothby-Katz}, although they happen to coincide when $q=2$.

In the Introduction we have already stated that we equip the $\C$-vector space of aperiodic sequences with the usual inner product: if $f,g$ are in this space, then $\ip{f}{g}=\sum_{j \in \Z} f_j \conj{g_j}$, so that $\ip{f}{f}=\normtwosq{f}$.
And similarly, for each $\ell \in \Z^+$, we equip the $\C$-vector space of periodic sequences of length $\ell$ with the usual inner product: if $f,g$ are in this space, then $\ip{f}{g}=\sum_{j \in \Zell} f_j \conj{g_j}$, so that $\ip{f}{f}=\normtwosq{f}$.

\subsection{Sequences from multiplicative characters}\label{Marlene}

Recall that if $p$ is a prime, then $\Fp$ is the finite field of order $p$, $\Fpu$ is the unit group of $\Fp$, and we let $\alpha_p$ be a fixed primitive element (generator of $\Fpu$).
A {\it multiplicative character} of $\Fp$ is a homomorphism from the group $\Fpu$ to the group $\C^*$.
The set of all multiplicative characters of $\Fp$ forms a cyclic group $\mchars$ of order $p-1$, where the group operation is multiplication of functions, i.e., $(\phi\chi)(a)=\phi(a)\chi(a)$ for $\phi,\chi\in\mchars$ and $a \in \Fpu$.
So if $\chi\in\mchars$, $a \in\Fpu$, and $k$ is an integer, then $\chi^k(a)=(\chi(a))^k$, and since the values of the characters lie on the complex unit circle (being finite order elements of the group $\C^*$), we note that $\chi^{-1}(a)=\overline{\chi(a)}$, so we introduce the notation $\overline{\chi}$ for $\chi^{-1}$, and call this character the {\it conjugate} of $\chi$.
The identity of the group $\mchars$ is the {\it trivial multiplicative character}, $\chi_0$, which maps every element of $\Fpu$ to $1$.
For each $p$, we let $\omega_p$ be the generator of $\mchars$ with $\omega_p(\alpha_p)=\exp(2\pi i/(p-1))$.
We extend any multiplicative character $\chi$ (including the trivial one) to have domain $\Fp$ by setting $\chi(0)=0$.

We are interested in periodic sequences that are linear combinations of multiplicative characters.
If $n \in \Z^+$, then a {\it character pattern of index $n$} is a periodic sequence of length $n$, i.e., a function $e\colon \Zn \to \C$.
If $p$ is a prime with $p \equiv 1 \pmod{n}$, then the {\it periodic sequence of length $p$ derived from character pattern $e$} is the periodic sequence $f\colon \Fp\to\C$ with
\begin{align*}
f_h = \sum_{j \in \Zn} e_j \omega_p^{(p-1) j/n}(h),
\end{align*}
so $f$ is a linear combination of the characters residing in the unique cyclic subgroup $\{\omega_p^{(p-1) j /n}: j \in \Zn\}$ of order $n$ in $\mchars$.
Note that each character in this subgroup has a constant value on a cyclotomic class of the form $\alpha_p^k\Fpn$.
Since character patterns of index $n$ are just periodic sequences of length $n$, which form a $\C$-inner product space, we can speak of individual character patterns as being {\it normalized} and of pairs of character patterns as being {\it orthogonal} or {\it orthonormal}.

It turns out that sequences derived from character patterns of index $n$ are identical to sequences derived from cyclotomic patterns of index $n$ via the Fourier transform.
If $f$ is a periodic sequence of length $\ell$, then the {\it Fourier transform of $f$}, written $\ft{f}$, is the periodic sequence of length $\ell$ with $\ft{f}_j$ (which means $(\ft{f})_j$) equal to $\ell^{-1} \sum_{k \in \Zell} \exp(-2\pi i j k/\ell) f_k$ for each $j \in \Zell$.
If $f$ is a periodic sequence of length $\ell$, then the {\it inverse Fourier transform of $f$}, written $\ift{f}$, is the periodic sequence of length $\ell$ with $\ift{f}_j$ (which means $(\ift{f})_j$) equal to $\sum_{k \in \Zell} \exp(2\pi i j k/\ell) f_k$ for each $j \in \Zell$.
It is well known that both $f\mapsto \ft{f}$ and $f\mapsto \ift{f}$ are $\C$-linear automorphisms of the $\C$-vector space of periodic sequences of length $\ell$, and are inverses of each other.
It is also well known that (up to scaling) the Fourier transform and its inverses are $l^2$ isometries: $\ip{\ft{f}}{\ft{g}}=\ell^{-1} \ip{f}{g}$ and $\ip{\ift{f}}{\ift{g}}=\ell \ip{f}{g}$.

The following basic result will be used to show that sequences derived from cyclotomic patterns are the same as sequences derived from character patterns.
\begin{lemma}\label{Donald}
Let $n\in\Z^+$, let $e$ be a character pattern of index $n$, let $p$ be a prime with $p\equiv 1 \pmod{n}$, and let $f$ be the periodic sequence of length $p$ derived from character pattern $e$.
Then $f_0=0$ and for every $k \in \Z$, $f_{\alpha_p^k}=\ift{e}_k$.
\end{lemma}
\begin{proof}
It is clear that $f_0=\sum_{j \in \Zn} e_j \omega_p^{(p-1) j/n}(0)=0$.
Let $k \in \Z$, and then
\begin{align*}
f_{\alpha_p^k}
& = \sum_{j \in \Zn} e_j  \omega_p^{(p-1) j/n}(\alpha_p^k) \\
& = \sum_{j \in \Zn} e_j  \exp(2\pi i/(p-1))^{(p-1) j k/n} \\
& = \sum_{j \in \Zn} e_j  \exp(2\pi i j k/n) \\
& = \ift{e}_k.\qedhere
\end{align*}
\end{proof}
This simple calculation has the following consequence.
\begin{proposition}\label{Wally}
Let $n\in\Z^+$, let $d$ be a cyclotomic pattern of index $n$, and let $e$ be a character pattern of index $n$.  Then the following are equivalent:
\begin{enumerate}[(i)]
\item\label{Gabriel} For every prime $p$ with $p\equiv 1 \pmod{n}$, the periodic sequence of length $p$ from $d$ is the same as the periodic sequence of length $p$ derived from character pattern $e$.
\item\label{Michael} There is some prime $p$ with $p\equiv 1 \pmod{n}$ such that the periodic sequence of length $p$ from $d$ is the same as the periodic sequence of length $p$ derived from character pattern $e$.
\item\label{Raphael} We have $d=\ift{e}$.
\item\label{Uriel} We have $e=\ft{d}$.
\end{enumerate}
\end{proposition}
\begin{proof}
If \eqref{Gabriel} holds, then \eqref{Michael} holds because the set of primes $p$ with $p\equiv 1 \pmod{n}$ is nonempty by Dirichlet's theorem on primes in arithmetic progression.

If \eqref{Michael} holds, then let $p$ be a prime with $p\equiv 1\pmod{n}$ and such that $f$ is a periodic sequence of length $p$ derived both from cyclotomic pattern $d$ and character pattern $e$.
If $k \in \Z$, then $f_{\alpha_p^k}=d_k$ because $f$ is from $d$, but on the other hand \cref{Donald} shows that $f_{\alpha_p^k}=\ift{e}_k$, which verifies that \eqref{Raphael} holds.

If \eqref{Raphael} holds, then \eqref{Uriel} holds by Fourier analysis.

If \eqref{Uriel} holds, then let $p$ be a prime with $p\equiv 1 \pmod{n}$ and let $f$ be the periodic sequence of length $p$ derived from character pattern $e$.
Then \cref{Donald} shows that $f_0=0$ and $f_{\alpha_p^k}=\ift{e}_k$ for every $k \in \Z$, so that $f_{\alpha_p^k}=d_k$ for every $k \in \Z$.
Since we read indices of $d$ modulo $n$, this means that for each $j \in \Z$, we have $f_j=d_k$ whenever $j \in \alpha_p^k \Fpn$, and so $f$ is the periodic sequence of length $p$ from $d$, so that \eqref{Gabriel} holds.
\end{proof}
This result means that the $l^2$ norms of cyclotomic patterns, character patterns, and the periodic and aperiodic sequences derived from them are all connected.
\begin{lemma}\label{Abigail}
Let $n \in \Z^+$, let $d$ and $d'$ be cyclotomic patterns of index $n$, and let $e=\ft{d}$ and $e'=\ft{d'}$.
Let $p$ be a prime with $p\equiv 1 \pmod{n}$ and let $f$ and $g$ be the periodic sequences of length $p$ derived from the cyclotomic patterns $d$ and $d'$, respectively (or equivalently, from the character patterns $e$ and $e'$, respectively).
Let $r \in \Z$.
Then
\[
\AC_{f^{(r)},g^{(r)}}(0)=\ip{f^{(r)}}{g^{(r)}}=\ip{f}{g}=\frac{p-1}{n} \ip{d}{d'} = (p-1) \ip{e}{e'}
\]
and
\[
\AC_{f^{(r)},f^{(r)}}(0)=\normtwosq{f^{(r)}}=\normtwosq{f}=\frac{p-1}{n} \normtwosq{d}=(p-1)\normtwosq{e}.
\]
The sequence $f$ is unimodularizable if and only if the cyclotomic pattern $d$ is unimodular, and when this is the case we have $\normtwosq{d}=n$, $\normtwosq{e}=1$, and $\AC_{f^{(r)},f^{(r)}}(0)=\normtwosq{f^{(r)}}=\normtwosq{f}=p-1$, and if $\uni{f}$ is a unimodularization of $f$, then $\AC_{\uni{f}^{(r)},\uni{f}^{(r)}}(0)=\normtwosq{\uni{f}^{(r)}}=\normtwosq{\uni{f}}=p$.
\end{lemma}
\begin{proof}
We have
\begin{align*}
\AC_{f^{(r)},g^{(r)}}(0)
& = \sum_{j \in \Z} f^{(r)}_j \conj{g^{(r)}_j} \\
& = \sum_{j=0}^{p-1} f_{j+r} \conj{g_{j+r}} \\
& = \sum_{h \in \Fp} f_h \conj{g_h} \\
& = \sum_{k=0}^{n-1} \sum_{h \in \alpha_p^k\Fpn} d_k d'_k \\
& = \frac{p-1}{n} \ip{d}{d'}\\
& = (p-1) \ip{e}{e'},
\end{align*}
where the first equality is from the definition of crosscorrelation, the second from the fact that $f^{(r)}$ (resp., $g^{(r)}$) is the rotation by $r$ of $f$ (resp., $g$), the third from the fact that $f$ and $g$ are periodic of length $p$, the fourth from the fact that $f$ and $g$ are derived from cyclotomic patterns $d$ and $d'$ respectively, the fifth from the fact that each cyclotomic class $\alpha_p^k\Fpn$ has cardinality $(p-1)/n$, and the sixth from the isometric property of the Fourier transform.
The first equality shows that $\AC_{f^{(r)},g^{(r)}}(0)$ equals $\ip{f^{(r)}}{g^{(r)}}$, the third shows that it equals $\ip{f}{g}$, and the fifth and sixth verify the other equalities we want for $\AC_{f^{(r)},g^{(r)}}(0)$.
We obtain the equalities for $\AC_{f^{(r)},f^{(r)}}(0)$ by specializing to the case where $f=g$.

The sequence $f$ is unimodularizable if and only if every $f_j$ with $j\in\Fpu$ is unimodular, which is true if and only if every $d_k$ with $k \in \Zn$ is unimodular, in which case $\normtwosq{d}=n$, and the norm values for $e$, $f^{(r)}$, and $f$ follow from the previous part of this lemma.
The effect of unimodularization is to change one term of $f$ (and thus one term of $f^{(r)}$) from $0$ to a unimodular complex number, so it increases the squared $l^2$ norm by $1$.
\end{proof}

Another property of that will become important in \cref{Bernard} is that that of balance.
\begin{lemma}\label{Greta}
Let $n \in \Z^+$, let $d$ be a cyclotomic pattern of index $n$, let $e=\ft{d}$, let $p$ be a prime with $p\equiv 1 \pmod{n}$, and let $f$ be the periodic sequence from $d$ (or equivalently, derived from the character plan $e$).
Then the following are equivalent:
\begin{enumerate}[(i).]
\item $d$ is balanced, i.e., $\sum_{j \in \Zn} d_j=0$,
\item $e_0=0$, and
\item $f$ is balanced, i.e., $\sum_{j \in \Fpu} f_j=0$.
\end{enumerate}
\end{lemma}
\begin{proof}
The equivalence of the first and second statements follows from the Fourier transform, and the equivalence of the first and third statements follows from the fact that $\sum_{j \in \Fpu} f_j=((p-1)/n) \sum_{j \in \Zn} d_j$ because every cyclotomic class contains $(p-1)/n$ elements.
\end{proof}

\subsection{Effects of unimodularization}\label{Estelle}

Recall the definitions of unimodularizability and unimodularization from \cref{Una}.
In this section, we show that unimodularization does not much affect our correlation measures, which allows us to prove asymptotic results on unimodularized sequences from the analogous results concerning sequences that have not been unimodularized, which are more mathematically tractable.

We first show that unimodularization cannot change peak correlation much.
\begin{lemma}\label{Ulrich}
Let $f$ and $g$ be unimodularizable periodic sequences, let $\uni{f}$ and $\uni{g}$ be unimodularizations of $f$ and $g$, respectively, and let $r,r' \in \Z$.
\begin{enumerate}[(i).]
\item For every $s \in \Z$, we have $|\AC_{f^{(r)},f^{(r)}}(s)-\AC_{\uni{f}^{(r)},\uni{f}^{(r)}}(s)| \leq 2$, so that $|\PSL(f^{(r)})-\PSL(\uni{f}^{(r)})| \leq 2$.
\item For every $s \in \Z$, we have $|\AC_{f^{(r)},g^{(r')}}(s)-\AC_{\uni{f}^{(r)},\uni{g}^{(r')}}(s)| \leq 2$, so that $|\PCC(f^{(r)},g^{(r')})-\PCC(\uni{f}^{(r)},\uni{g}^{(r')})| \leq 2$.
\end{enumerate}
\end{lemma}
\begin{proof}
The statements about $\PSL$ and $\PCC$ follow immediately from the inequalities on correlation at a particular shift $s$, and of these, the inequality in the first statement follows from the one in the second statement when $f=g$ and $r=r'$.
We note that
\[
\AC_{f^{(r)},g^{(r')}}(s)-\AC_{\uni{f}^{(r)},\uni{g}^{(r')}}(s) = \sum_{j \in \Z} \left(f^{(r)}_{j+s} \conj{g^{(r')}_j} - \uni{f}^{(r)}_{j+s} \conj{\uni{g}^{(r')}_j}\right).
\]
Since $f^{(r)}$ and $\uni{f}^{(r)}$ (resp., $g^{(r')}$ and $\uni{g}^{(r')}$) differ in only one term, there is a set of integers $K$ with $|K|\leq 2$ such that the summand $f^{(r)}_{j+s} \conj{g^{(r')}_j} - \uni{f}^{(r)}_{j+s} \conj{\uni{g}^{(r')}_j}=0$ for all $j \not\in K$, but when $j \in K$, then this summand is equal to $-\uni{f}^{(r)}_{j+s} \conj{\uni{g}^{(r')}_j}$ (which has magnitude at most $1$, since $\uni{f}^{(r)}$ and $\uni{g}^{(r')}$ are unimodular).
\end{proof}
This immediately implies the analogous result for codebooks.
\begin{corollary}\label{Susan}
Let $P$ be a unimodularizable periodic codebook.
For each $f \in P$, let $\uni{f}$ be a unimodularization of $f$ and $r_f \in \Z$.
Let $\uni{P}=\{\uni{f}: f \in P\}$, $A=\{f^{(r_f)}: f\in P\}$, and $\uni{A}=\{\uni{f}^{(r_f)}: f \in P\}$, and suppose that $|A|=|\uni{A}|=|P|$.
Then $|\GUC(A)-\GUC(\uni{A})|\leq 2$.
\end{corollary}
\begin{proof}
This follows immediately from \cref{Ulrich} once we impose the condition that $|A|=|\uni{A}|=|P|$.
We need this assumption: if two periodic sequences in $P$ (resp., $\uni{P}$) give rise to the same aperiodic sequence in $A$ (resp., $\uni{A}$) upon rotation, then the crosscorrelations between those two identical aperiodic sequences should not be counted when determining the $\GUC$ of $A$ (resp., $\uni{A}$).
\end{proof}
The consequence of this is that unimodularization does not influence the answer to the question as to whether or not a family of uniform-length codebooks has well regulated growth of $\GUC$.
\begin{corollary}\label{Melanie}
Let $\{P_\iota\}_{\iota \in I}$ be an infinite family of uniform-length, unimodularizable periodic codebooks.
For each $\iota \in I$, and for each $f \in P_\iota$, let $\uni{f}$ be a unimodularization of $f$ and $r_f \in \Z$, and let $\uni{P}_\iota=\{\uni{f}: f \in P_\iota\}$, $A_\iota=\{f^{(r_f)}: f\in P_\iota\}$, and $\uni{A}=\{\uni{f}^{(r_f)}: f \in P_\iota\}$, and suppose that $|A_\iota|=|\uni{A}_\iota|=|P_\iota|$.
Then $\SDC(A_\iota)=\len(P_\iota)-1$ and $\SDC(\uni{A}_\iota)=\len(P_\iota)$ for every $\iota \in I$, so that $\{\SDC(A_\iota):\iota \in I\}$ is unbounded if and only if $\{\SDC(\uni{A}_\iota):\iota \in I\}=\{\len(P_\iota):\iota \in I\}$ is unbounded.
If these sets are unbounded, then $\{A_\iota\}_{\iota \in I}$ has well regulated growth of $\GUC$ if and only if $\{\uni{A}_\iota\}_{\iota\in I}$ does.
\end{corollary}
\begin{proof}
For $\iota \in I$, every sequence in $A_\iota$ (resp., $\uni{A}_\iota$) has autocorrelation at shift zero equal to $\len(P_\iota)-1$ (resp., $\len(P_\iota)$), so that $\SDC(A_\iota)=\len(P_\iota)-1$ and $\SDC(\uni{A}_\iota)=\len(P_\iota)$.
Since $\GUC(A_\iota)$ and $\GUC(\uni{A}_\iota)$ differ by at most $2$ by \cref{Susan}, it is clear that, for any real $\epsilon > 0$, the ratios $\GUC(A_\iota)/(\SDC(A_\iota))^{1/2+\epsilon}$ and $\GUC(\uni{A}_\iota)/(\SDC(\uni{A}_\iota))^{1/2+\epsilon}$ have the same asymptotic behavior as $\SDC(A_\iota)$ and $\SDC(\uni{A}_\iota)$ tend to infinity.
\end{proof}

We shall also show that unimodularization does not affect the asymptotic crosscorrelation demerit factor of a family of codebooks.
We first provide a technical lemma.
\begin{lemma}\label{Judith}
Let $\ell >1$, let $f$ and $g$ be unimodularizable periodic sequences of length $\ell$, and let $\uni{f}$ and $\uni{g}$ be unimodularizations of $f$ and $g$, respectively.
For every $r,r' \in \Z$,
\begin{enumerate}[(i).]
\item\label{Valerie} We have
\[
\left|\sqrt{\CDF(\uni{f}^{(r)},\uni{g}^{(r')})}-\sqrt{\CDF(f^{(r)},g^{(r')})}\right| < \frac{\sqrt{\CDF(f^{(r)},g^{(r')})}}{\ell} + \frac{2\sqrt{\ell}}{\ell}.
\]
\item\label{Timothy} We have
\[
\left|\sqrt{\CDF(\uni{f}^{(r)},\uni{g}^{(r')})} -\sqrt{\CDF(f^{(r)},g^{(r')})}\right| < \frac{\sqrt{\CDF(\uni{f}^{(r)},\uni{g}^{(r')})}}{\ell-1} + \frac{2\sqrt{\ell}}{\ell-1}.
\]
\end{enumerate}
\end{lemma}
\begin{proof}
We identify each aperiodic sequence $a\colon \Z\to \C$ with the Laurent polynomial $a(z)=\sum_{j \in \Z} a_j z^j$.
Then we note that the $L^2$ norm this Laurent polynomial on the complex unit circle
\[
\frac{1}{2\pi} \left(\int_{0}^{2\pi} |a(\exp(i\theta))|^2 d\theta \right)^{1/2}
\]
is exactly equal to our usual $l^2$ norm, $\normtwo{a}$.
In \cite[eq.~(13)]{Katz}, it is shown that if $a$ and $b$ are aperiodic sequences, then
\begin{equation}\label{Celeste}
\CDF(a,b)=\normtwosq{a b}/(\normtwosq{a} \normtwosq{b}).
\end{equation}

Two applications of the triangle inequality yield
\begin{align*}
\big|\normtwo{\uni{f}^{(r)} \uni{g}^{(r')}}-\normtwo{f^{(r)} g^{(r')}}\big|
& \leq \normtwo{\uni{f}^{(r)} \uni{g}^{(r')}-f^{(r)} g^{(r')}} \\
& \leq  \normtwo{(\uni{f}^{(r)}-f^{(r)}) \uni{g}^{(r')}} + \normtwo{f^{(r)}(\uni{g}^{(r')}-g^{(r')})},
\end{align*}
but also note that $\uni{f}^{(r)}(z)-f^{(r)}(z)$ (resp., $\uni{g}^{(r')}(z)-g^{(r')}(z)$) is a single monomial with a unimodular coefficient.
Multiplication of a Laurent polynomial by a single monomial with a unimodular coefficient does not change its absolute value anywhere on the unit circle, so
\begin{equation}\label{Anastasia}
\big|\normtwo{\uni{f}^{(r)} \uni{g}^{(r')}}-\normtwo{f^{(r)} g^{(r')}}\big| \leq \normtwo{\uni{g}^{(r')}} + \normtwo{f^{(r)}} =\sqrt{\ell}+\sqrt{\ell-1} < 2\sqrt{\ell},
\end{equation}
since $\uni{g}^{(r')}$ is unimodular of length $\ell$, and $f^{(r)}$ has $\ell-1$ nonvanishing terms, all unimodular.

If we divide \eqref{Anastasia} through by $\normtwo{\uni{f}^{(r)}}\normtwo{\uni{g}^{(r')}}=\ell$ and use \eqref{Celeste} to convert this into inequalities involving crosscorrelation demerit factors, we obtain
\[
\left|\sqrt{\CDF(\uni{f}^{(r)},\uni{g}^{(r')})} - \sqrt{\CDF(f^{(r)},g^{(r')})} \left(1-\frac{1}{\ell}\right)\right| < \frac{2\sqrt{\ell}}{\ell},
\]
from which one obtains our first result, \eqref{Valerie}.
On the other hand, if we divide \eqref{Anastasia} through by $\normtwo{f}\normtwo{g}=\ell-1$ and use \eqref{Celeste} to convert this into inequalities involving crosscorrelation demerit factors, we obtain
\[
\left|\sqrt{\CDF(\uni{f}^{(r)},\uni{g}^{(r')})} \left(1+\frac{1}{\ell-1}\right) - \sqrt{\CDF(f^{(r)},g^{(r')})}\right| < \frac{2\sqrt{\ell}}{\ell-1},
\]
from which one obtains our second result, \eqref{Timothy}.
\end{proof}
Now we show that unimodularization does not affect the asymptotic crosscorrelation demerit factor of a family of codebooks.
\begin{lemma}\label{Karl}
Let $\{P_\iota\}_{\iota \in I}$ be an infinite family of uniform-length, unimodularizable periodic codebooks with the same number of sequences in each codebook.
For each $\iota \in I$, let $\ell_\iota=\len(P_\iota)$, and suppose that $\{\ell_\iota: \iota \in I\}$ is unbounded.
For each $\iota$, let $\uni{P}_\iota=\{\uni{f}: f \in P_\iota\}$, let $A_\iota=\{f^{(r_f)}: f \in P_\iota\}$, and let $\uni{A}_\iota=\{\uni{f}^{(r_f)}: f \in P_\iota\}$, where for each $f \in P_\iota$, we use $\uni{f}$ to denote some unimodularization of $f$ and $r_f$ to denote some integer, and suppose that $|P_\iota|=|\uni{P}_\iota|=|A_{\iota}|=|\uni{A}_\iota|$.
Then as $\ell_\iota\to\infty$, the quantity $\CDF(A_\iota)$ tends to a real number if and only if $\CDF(\uni{A}_\iota)$ does, in which case they tend to the same limit.
\end{lemma}
\begin{proof}
Since this is an asymptotic result in the limit as $\ell_\iota$ tends to infinity, we may, without loss of generality, assume that $\ell_\iota > 1$ for all $\iota \in I$.
Let $N$ denote the number of sequences in each of the codebooks $P_\iota$ (and $\uni{P}_\iota$ and $A_\iota$ and $\uni{A}_\iota$).

Suppose that as $\ell_\iota\to\infty$, the quantity $\CDF(A_\iota)$ tends to a real number $C$.
This means that there is some positive $L$ such that whenever $\ell_\iota \geq L$, we have $\CDF(A_\iota) < C+1$, and thus $\CDF(f^{(r_f)},g^{(r_g)}) < N^2(C+1)$ for every $f^{(r_f)},g^{(r_g)} \in A_\iota$; we restrict our attention to values of $\iota$ with $\ell_\iota \geq L$ for the rest of this paragraph.
Therefore, if $f^{(r_f)},g^{(r_g)} \in A_\iota$ and $\uni{f}^{(r_f)}$ and $\uni{g}^{(r_g)}$ are the corresponding sequences in $\uni{A}_\iota$, we use \cref{Judith}\eqref{Valerie} to obtain
\[
\left|\sqrt{\CDF(\uni{f}^{(r_f)},\uni{g}^{(r_g)})}-\sqrt{\CDF(f^{(r_f)},g^{(r_g)})}\right| < \frac{N \sqrt{C+1}}{\ell_\iota} + \frac{2\sqrt{\ell_\iota}}{\ell_\iota}.
\]
Thus, by continuity of the function $x\mapsto x^2$ on $\R$, for any $\epsilon > 0$, there is some $L_\epsilon$ such that if $\iota \in I$ with $\ell_\iota \geq L_\epsilon$, then every pair of sequences $f^{(r_f)},g^{(r_g)} \in A_\iota$ with corresponding sequences $\uni{f}^{(r_f)}, \uni{g}^{(r_g)} \in \uni{A}_\iota$ has
\[
\left|\CDF(\uni{f}^{(r_f)},\uni{g}^{(r_g)})-\CDF(f^{(r_f)},g^{(r_g)})\right| < \epsilon,
\] and so by the triangle inequality $|\CDF(\uni{A}_\iota)-\CDF(A_\iota)| < \epsilon$, and so we conclude that $\CDF(\uni{A}_\iota)$ tends to the same limit as $\CDF(A_\iota)$ does when $\ell_\iota \to \infty$.

If $\CDF(\uni{A}_\iota)$ tends to some real number as $\ell_\iota\to\infty$, then the proof that $\CDF(A_\iota)$ tends to the same number proceeds along very much the same lines as the proof in the previous paragraph, but uses part \eqref{Timothy} of \cref{Judith} in place of part \eqref{Valerie}.
\end{proof}

\section{Growth of greatest undesirable correlation for codebooks from cyclotomic plans}\label{Bernard}

In this section, we prove our first main result, \cref{Penelope}.
To set the stage, we introduce additive characters and Gauss sums in \cref{Gerald}.
Then we bound the aperiodic crosscorrelation of sequences derived from character patterns in \cref{William}.
Then we prove \cref{Penelope} in \cref{Margaret}.

\subsection{Additive characters and Gauss sums}\label{Gerald}

If $p$ is a prime, then an {\it additive character of $\Fp$} is a homomorphism from the additive group $\Fp$ into $\C^*$.
Let $\achars$ be the set of all additive characters of $\Fp$, which forms a group under pointwise multiplication: $(\kappa\lambda)(a)=\kappa(a)\lambda(a)$ for $\kappa,\lambda \in \achars$ and $a\in\Fp$.
The {\it canonical additive character of $\Fp$} is $\epsilon_p\colon \Fp \to \C$ with $\epsilon_p(j)=\exp(2\pi i j/p)$ for every $j \in \Fp$.
The canonical additive character has order $p$ and generates $\achars$, so $\achars=\{\epsilon_p^0,\ldots,\epsilon_p^{p-1}\}$.
Then the map $j\mapsto \epsilon_p^j$ gives an isomorphism of cyclic groups from the additive group of $\Fp$ to the group $\achars$.
Since character values lie on the complex unit circle, if $\kappa\in\achars$, then we write $\conj{\kappa}$ to mean the inverse of $\kappa$ in the group $\achars$, so that $\kappa^{-1}(a)=\conj{\kappa}(a)=\conj{\kappa(a)}$ for every $a \in \Fp$.

If $\psi \in \achars$ and $\chi \in \mchars$, then the {\it Gauss sum for $\psi$ and $\chi$} is
\[
G(\psi,\chi)=\sum_{x \in \Fpu} \psi(x) \chi(x).
\]
If $a \in \Fp$, we use the shorthand $G_a(\chi)$ to mean $G(\epsilon_p^a,\chi)$ and the shorthand $G(\chi)$ to mean $G_1(\chi)=G(\epsilon_p,\chi)$.
We record some useful facts about Gauss sums.
\begin{lemma}\label{Gertrude}
Let $p$ be a prime, let $a \in \Fp$, let $\chi\in\mchars$, and recall that $\chi_0$ denotes the trivial multiplicative character.  Then
\begin{enumerate}[(i).]
\item\label{Alphonse} $G_0(\chi_0)=p-1$,
\item\label{Bogdan} $G_0(\chi)=0$ if $\chi\not=\chi_0$,
\item\label{Clarence} $G_a(\chi_0)=-1$ if $a\not=0$,
\item\label{Dorothy} $|G_a(\chi)|=\sqrt{p}$ if $a\not=0$ and $\chi\not=\chi_0$, 
\item\label{Elaine} $G_a(\chi)=\conj{\chi}(a) G(\chi)$ if $a\not=0$ or if $\chi\not=\chi_0$, and
\item\label{Felicia} $G(\conj{\chi})=\chi(-1) \conj{G(\chi)}$.
\end{enumerate}
\end{lemma}
\begin{proof}
See \cite[Theorem 5.11, 5.12(i),(iii)]{Lidl-Niederreiter}.
\end{proof}
Gauss sums provide a way of expressing multiplicative characters as linear combinations of additive characters.
\begin{lemma}\label{Horace}
If $\chi\in\mchars$ and $b \in \Fp$, then
\[
\chi(b)=\frac{1}{p} \sum_{a \in \Fp} G_a(\chi) \conj{\epsilon_p^a}(b).
\]
\end{lemma}
\begin{proof}
See \cite[Lemma 8]{Katz}.
\end{proof}

\subsection{Bounds on peak correlation of sequences from character patterns}\label{William}

Our results depend on two character sum bounds.
The first is Weil's bound \cite{Weil} on character sums with polynomial arguments, and the other is Sarwate's version \cite[Lemma 1]{Sarwate-1984-Upper} of a bound of Vinogradov.
We prove a slight generalization of Sarwate's bound for use in this paper.
\begin{lemma}\label{Samantha}
Let $N \in \Z^+$ and let $B$ be any finite set of consecutive integers.
Then $\sum_{a=1}^{N-1} \left|\sum_{b \in B} \exp(2\pi i a b/N) \right| < \frac{2 N}{\pi} \log\left(\frac{4 N}{\pi}\right).$
\end{lemma}
\begin{proof}
Let $b_0 \in \Z$ and $\ell \in \N$ be such that $B=\{b \in \Z: b_0 \leq b < b_0+\ell\}$.
Then for each $a \in \{1,\ldots,N-1\}$, we have
\[
\sum_{b \in B} \exp(2\pi i a b/N) = \frac{\exp(2\pi i a b_0/N)- \exp(2\pi i a (b_0+\ell)/N)}{1-\exp(2\pi i a/N)},
\]
so that if $m$ is the least nonnegative integer with $m \equiv \ell \pmod{N}$, we have
\[
\left|\sum_{b \in B} \exp(2\pi i a b/N)\right| = \left|\frac{1-\exp(2\pi i a m/N)}{1-\exp(2\pi i a/N)}\right|,
\]
so that $|\sum_{b \in B} \exp(2\pi i a b/N)|=|\sum_{b=0}^{m-1} \exp(2\pi i a b/N)|$, and so the sum we wish to bound equals $\sum_{a=1}^{N-1} \left|\sum_{b=0}^{m-1} \exp(2\pi i a b/N) \right|$.
If $m=0$, then the inner sums are empty, and so the desired bound is trivial; otherwise, the last double sum is $N$ times the sum that Sarwate \cite[Lemma 1]{Sarwate-1984-Upper} calls $\Gamma_{N-m,N}$ and upper bounds by $(2/\pi)\log(4 N/\pi)$.
\end{proof}
Now we are ready to prove bounds on the correlation of aperiodic sequences obtained by rotating periodic sequences given by single multiplicative characters.
\begin{lemma}\label{Esmeralda}
Let $p$ be a prime, let $\phi, \chi \in \mchars$, let $f$ and $g$ be periodic sequences of length $p$ given by $f_j=\phi(j)$ and $g_j=\chi(j)$ for all $j \in \Fp$, and let $r,r',s \in \Z$ with $|s|<p$.
\begin{enumerate}[(i)]
\item\label{apple} If $\phi$ and $\chi$ are both the trivial character, then $\AC_{f^{(r)},g^{(r')}}(s) \in \{p-|s|,p-|s|-1, p-|s|-2\}$.
\item\label{banana} If precisely one of $\phi$ or $\chi$ is trivial, we have $|\AC_{f^{(r)},g^{(r')}}(s)| \leq  1 + \frac{2}{\pi} \sqrt{p} \log\left(\frac{4 p}{\pi}\right)$.
\item\label{cherry} If $\phi$ and $\chi$ are the same nontrivial character and $s \equiv r'-r \pmod{p}$, then $\AC_{f^{(r)},g^{(r')}}(s) \in \{p-|s|,p-|s|-1\}$.
\item\label{date} If $\phi$ and $\chi$ are the same nontrivial character and $s \not\equiv r'-r \pmod{p}$, then $|\AC_{f^{(r)},g^{(r')}}(s)| \leq 2 \sqrt{p} + \frac{4}{\pi} \sqrt{p} \log\left(\frac{4 p}{\pi}\right)$.
\item\label{elderberry} If $\phi, \chi$ are distinct nontrivial characters, then $|\AC_{f^{(r)},g^{(r')}}(s)| \leq 2 \sqrt{p} + \frac{4}{\pi} \sqrt{p} \log\left(\frac{4 p}{\pi}\right)$.
\end{enumerate}
Therefore, if we define
\[
B=\begin{cases}
p-|s| & \text{if $\phi=\chi=\chi_0$,}\\
p-|s| & \text{if $\phi=\chi$ and $s\equiv r'-r \pmod{p}$,}\\
0 & \text{otherwise,}
\end{cases}
\]
then
\[
|\AC_{f^{(r)},g^{(r')}}(s)-B| \leq 2 \sqrt{p} + \frac{4}{\pi} \sqrt{p} \log\left(\frac{4 p}{\pi}\right).
\]
\end{lemma}
\begin{proof}
We have
\begin{equation}\label{Henry}
\begin{aligned}
\AC_{f^{(r)},g^{(r')}}(s)
& = \sum_{j \in \Z} f^{(r)}_{j+s} \conj{g^{(r')}_j} \\
& = \sum_{j \in J} \phi(j+s+r) \conj{\chi}(j+r') 
\end{aligned}
\end{equation}
where $J=\{0,1,\ldots,p-1-s\}$ if $s \geq 0$, but $J=\{-s,-s+1,\ldots,p-1\}$ when $s \leq 0$.
(The first equality is due to the definition of aperiodic autocorrelation and the second is by the definitions of $f^{(r)}$ and $g^{(r')}$.)

If $\phi$ and $\chi$ are both trivial, then the summand $\phi(j+s+r) \conj{\chi}(j+r')$ of \eqref{Henry} is $1$ unless $j \equiv -(s+r) \pmod{p}$ or $j\equiv -r' \pmod{p}$, in which case the summand instead becomes $0$.  Since $J$ is a set of $p-|s|$ consecutive integers, each of these congruences can be satisfied for at most one $j$ in $J$, and so $\AC_{f^{(r)},g^{(r')}}(s)$ is in $\{|J|, |J|-1, |J|-2\}$, and so \eqref{apple} follows.

If $\phi=\chi$ and $s\equiv r'-r \pmod{p}$, \eqref{Henry} yields $\AC_{f^{(r)},g^{(r')}}(s)  = \sum_{j \in J} \phi(j+r') \conj{\phi}(j+r')$, whose $j$th term is $1$ unless $j \equiv -r'\pmod{p}$, in which case the $j$th term is $0$.
Since $J$ is a set of $p-|s|$ consecutive integers, this means that $j \equiv -r' \pmod{p}$ occurs for at most one $j$ in $J$, whence \eqref{cherry} follows.

Now suppose that $\phi$ is trivial and $\chi$ is not, and consider \eqref{Henry}.
We have $\phi(j+s+r)=1$ except when $j\equiv -(s+r)\pmod{p}$, which can happen for at most one $j \in J$, and if this does happen, the $j$th summand of \eqref{Henry} is $0$ rather than $\conj{\chi}(j+r')$.  So there is some $E \in \C$ with $|E| \leq 1$ such that
\begin{align*}
\AC_{f^{(r)},g^{(r')}}(s)
& = E+ \sum_{j \in J} \conj{\chi}(j+r') \\
& = E+ \sum_{j \in J} \frac{1}{p} \sum_{b \in \Fp} \conj{G_b(\chi)} \epsilon_p^b(j+r') && \text{by \cref{Horace}} \\
& = E+ \frac{\conj{G(\chi)}}{p} \sum_{b \in \Fp} \chi(b) \sum_{j \in J} \epsilon_p^b(j+r')  && \text{by \cref{Gertrude}\eqref{Elaine}},
\end{align*}
so we may use the triangle inequality to obtain
\begin{align*}
\left|\AC_{f^{(r)},g^{(r')}}(s)\right|
& \leq 1 + \frac{|G(\chi)|}{p} \sum_{b \in \Fpu} \left|\sum_{j \in J} \epsilon_b^p(j+r')\right| \\
& = 1 + \frac{1}{\sqrt{p}} \sum_{b \in \Fpu} \left|\sum_{j \in J} \epsilon_b^p(j+r')\right| && \text{by \cref{Gertrude}\eqref{Dorothy}} \\
& = 1 + \frac{2}{\pi} \sqrt{p} \log\left(\frac{4 p}{\pi}\right) && \text{by \cref{Samantha}.}
\end{align*}

On the other hand, if $\chi$ is trivial and $\phi$ is not, then \[|\AC_{f^{(r)},g^{(r')}}(s)|=|\conj{\AC_{g^{(r')},f^{(r)}}(-s)}| \leq 1 + \frac{2}{\pi} \sqrt{p} \log\left(\frac{4 p}{\pi}\right)\] by the previous paragraph; this finishes the proof of \eqref{banana}.

For the rest of this proof we assume that both $\phi$ and $\chi$ are nontrivial, and use \cref{Horace} in \eqref{Henry} to obtain
\[
\AC_{f^{(r)},g^{(r')}}(s) = \sum_{j\in J} \frac{1}{p^2} \sum_{a, b \in \Fp} G_a(\phi) \conj{\epsilon_p^a(j+s+r)} \conj{G_b(\chi) \conj{\epsilon_p^b}(j+r')}.
\]
Then we use \cref{Gertrude}\eqref{Elaine} to obtain
\begin{align*}
\AC_{f^{(r)},g^{(r')}}(s)
& =  \frac{G(\phi) \conj{G(\chi)}}{p^2} \sum_{a, b \in \Fp}  \conj{\phi}(a) \chi(b) \epsilon_p(b r'-a r- a s)  \sum_{j\in J} \epsilon_p((b-a)j) \\
& =  \frac{G(\phi) \conj{G(\chi)}}{p^2}  \sum_{c \in \Fp}  \epsilon_p(c r')   \sum_{j\in J} \epsilon_p(c j) \sum_{a \in \Fp} \epsilon_p^{r'-r-s}(a) \conj{\phi}(a) \chi(a+c),
\end{align*}
where we have reparameterized the sum of $(a,b)$ over $\Fp\times\Fp$ using the bijection $(a,c) \mapsto (a,a+c)$ in the second equality.
The group $\mchars$ is cyclic of order $p-1$ with generator $\omega_p$.
Write $\phi=\omega_p^u$ and $\chi=\omega_p^v$ for some integers $u,v \in \{1,\ldots,p-2\}$.
Then the sum over $a$ in the last expression becomes
\[
W_c=\sum_{a \in \Fp} \epsilon_p^{r'-r-s}(a) \omega_p(a^{p-1-u} (a+c)^v),
\]
so that
\[
\AC_{f^{(r)},g^{(r')}}(s) = \frac{G(\phi) \conj{G(\chi)}}{p^2}  \sum_{c \in \Fp}  \epsilon_p(c r')   W_c \sum_{j\in J} \epsilon_p(c j).
\]
Weil's bound (or, in some cases, direct evaluation) tells us that
\[
|W_c| \leq
\begin{cases}
2\sqrt{p} & \text{if $r'-r\not\equiv s \pmod{p}$ and $c\not=0$;} \\
\sqrt{p} & \text{if $r'-r\equiv s \pmod{p}$ and $c\not=0$;} \\
\sqrt{p} & \text{if $r'-r\not\equiv s \pmod{p}$, $c=0$, and $u\not=v$;} \\
0 & \text{if $r'-r\equiv s \pmod{p}$, $c=0$, and $u\not=v$;} \\
1 & \text{if $r'-r\not\equiv s \pmod{p}$, $c=0$, and $u=v$; and} \\
p-1 & \text{if $r'-r\equiv s \pmod{p}$, $c=0$, and $u=v$.} 
\end{cases}
\]
Therefore, if $r'-r \not\equiv s \pmod{p}$ or $\phi\not=\chi$ (i.e., $u\not=v$), we have
\begin{align*}
|\AC_{f^{(r)},g^{(r')}}(s)|
& \leq \frac{|G(\phi)| |G(\chi)|}{p^2} \sum_{c \in \Fp} 2 \sqrt{p} \left|\sum_{j\in J} \epsilon_p(c j)\right|  \\
& = \frac{2}{\sqrt{p}} \sum_{c \in \Fp} \left|\sum_{j\in J} \epsilon_p(c j)\right| && \text{by Lemma \ref{Gertrude}\eqref{Dorothy}} \\
& \leq 2\sqrt{p} + \frac{2}{\sqrt{p}} \sum_{c \in \Fpu} \left|\sum_{j\in J} \epsilon_p(c j)\right| && \text{since $|J|=p-|s|\leq p$} \\
& \leq 2 \sqrt{p} + \frac{4}{\pi} \sqrt{p} \log\left(\frac{4 p}{\pi}\right) && \text{by \cref{Samantha}}.
\end{align*}
This finishes the proof of \eqref{date} and \eqref{elderberry}.
\end{proof}
\begin{remark}\label{Rebecca}
\cref{Esmeralda} gives an upper bound of $2 \sqrt{p} + \frac{4}{\pi} \sqrt{p} \log\left(\frac{4 p}{\pi}\right)$ on the peak sidelobe level of any rotation of the sequence given by the quadratic character $\eta$ (i.e., the Legendre symbol).
The Legendre sequence is usually defined by first unimodularizing: one replaces the $\eta(0)=0$ entry with a $1$; \cref{Ulrich} shows that this (or any other unimodularization) will produce a unimodularized instance with $\PSL$ upper bounded by $2+2 \sqrt{p} + \frac{4}{\pi} \sqrt{p} \log\left(\frac{4 p}{\pi}\right)$, which is less than the bound of $2+18 \sqrt{p} \log p$ that comes from the paper of Mauduit and S\'ark{\"o}zy \cite[Cor.~1]{Mauduit-Sarkozy}.
\end{remark}
Now we use the bound of \cref{Esmeralda} to bound correlation values for sequences from character patterns.
\begin{proposition}\label{Edna}
Let $n \in \Z^+$ and let $e$ and $e'$ be character patterns of index $n$.
Let $p$ be a prime with $p\equiv 1 \pmod{n}$, and let $f$ and $g$ be periodic sequences of length $p$ derived from character patterns $e$ and $e'$, respectively.
For $r,r',s \in \Z$ with  $|s| < p$, let
\[
B = \begin{cases}
(p-|s|) \ip{e}{e'} & \text{if $s \equiv r'-r \pmod{p}$,} \\
(p-|s|) e_0 \conj{e'_0} & \text{otherwise.}
\end{cases}
\]
Then
\[
|\AC_{f^{(r)},g^{(r')}}(s)-B| \leq \normone{e} \normone{e'} \left(2 \sqrt{p} + \frac{4}{\pi} \sqrt{p} \log\left(\frac{4 p}{\pi}\right)\right).
\]
\end{proposition}
\begin{proof}
Let $\phi=\omega_p^{(p-1)/n}$, so that $\{\phi^j: j \in \Zn\}$ is the unique cyclic subgroup of order $n$ in $\mchars$.
We have
\begin{align*}
\AC_{f^{(r)},g^{(r')}}(s)
& = \sum_{j \in \Z} f^{(r)}_{j+s} \conj{g^{(r')}_j} \\
& = \sum_{j \in J} \sum_{u \in \Zn} e_u \phi^u(j+s+r) \conj{\sum_{v \in \Zn} e'_v \phi^v(j+r')},
\end{align*}
where $J=\{0,1,\ldots,p-1-s\}$ if $s \geq 0$, but $J=\{-s,-s+1,\ldots,p-1\}$ when $s \leq 0$.
Thus, we have
\[
\AC_{f^{(r)},g^{(r')}}(s) = \sum_{u,v \in \Zn} e_u \conj{e'_v}  \AC_{(\phi^u)^{(r)},(\phi^v)^{(r')}}(s).
\]
For each $(u,v) \in \Zn\times\Zn$, we define
\[
B_{u,v}=\begin{cases}
p-|s| & \text{if $u=v=0$,}\\
p-|s| & \text{if $u=v$ and $s\equiv r'-r \pmod{p}$,}\\
0 & \text{otherwise,}
\end{cases}
\]
and
\[
E_{u,v} =  \AC_{(\phi^u)^{(r)},(\phi^v)^{(r')}}(s) - B_{u,v}.
\]
Then \cref{Esmeralda} shows that $|E_{u,v}| \leq 2 \sqrt{p} + \frac{4}{\pi} \sqrt{p} \log\left(\frac{4 p}{\pi}\right)$ for every $u,v$.
So
\begin{align*}
\AC_{f^{(r)},g^{(r')}}(s)
& = \sum_{u,v \in \Zn} e_u \conj{e'_v}  (B_{u,v} + E_{u,v}) \\
& = \sum_{u \in \Zn} e_u \conj{e'_u} B_{u,u}  + \sum_{u,v \in \Zn} e_u \conj{e'_v} E_{u,v} \\
& = B  + \sum_{u,v \in \Zn} e_u \conj{e'_v} E_{u,v},
\end{align*}
and therefore
\begin{align*}
|\AC_{f^{(r)},g^{(r')}}(s)-B|
& \leq \sum_{u,v \in \Zn} |e_u| |\conj{e'_v}| |E_{u,v}| \\
& \leq \normone{e} \normone{e'} \left(2 \sqrt{p} + \frac{4}{\pi} \sqrt{p} \log\left(\frac{4 p}{\pi}\right)\right).\qedhere
\end{align*}
\end{proof}
We specialize this result for autocorrelation.
\begin{corollary}\label{Mildred}
Let $n \in \Z^+$ and let $e$ be a character pattern of index $n$.
Let $p$ be a prime with $p \equiv 1 \pmod{n}$ and let $f$ be the periodic sequence of length $p$ from $e$.
For $r,s \in \Z$ with  $|s| < p$, let
\[
B = \begin{cases}
p \normtwosq{e} & \text{if $s=0$,} \\
(p-|s|) |e_0|^2 & \text{otherwise.}
\end{cases}
\]
Then
\[
|\AC_{f^{(r)},f^{(r)}}(s)-B| \leq \normonesq{e}  \left(2 \sqrt{p} + \frac{4}{\pi} \sqrt{p} \log\left(\frac{4 p}{\pi}\right)\right).
\]
\end{corollary}
We now use \cref{Edna} and \cref{Mildred} to give bounds on the peak correlation measures for individual instances or pairs of instances from cyclotomic patterns.
\begin{lemma}\label{Martha}
Let $n \in \Z^+$ let $d$ and $d'$ be cyclotomic patterns of index $n$, and let $p$ be a prime with $p\equiv 1\pmod{n}$.
Let $f$ (resp., $g$) be a $p$-instance of $d$ (resp., $d'$).
If $d$ (resp., $d'$) is unimodular, let $\uni{f}$ (resp., $\uni{g}$) be a unimodularized $p$-instance of $d'$.
\begin{enumerate}[(i).]
\item\label{Barbara} If $d$ is balanced, then
\[
\PSL(f) \leq \normonesq{\ft{d}}  \left(2 \sqrt{p} + \frac{4}{\pi} \sqrt{p} \log\left(\frac{4 p}{\pi}\right)\right)
\]
and if $d$ is also unimodular, then 
\[
\PSL(\uni{f}) \leq \normonesq{\ft{d}}  \left(2 \sqrt{p} + \frac{4}{\pi} \sqrt{p} \log\left(\frac{4 p}{\pi}\right)\right)+2.
\]
\item\label{Ursula} If $d$ is not balanced, then $\ft{d}_0\not=0$ and
\[
\PSL(f) \geq |\ft{d}_0|^2 (p-1) - \normonesq{\ft{d}} \left(2 \sqrt{p} + \frac{4}{\pi} \sqrt{p} \log\left(\frac{4 p}{\pi}\right)\right).
\]
and if $d$ is also unimodular, then 
\[
\PSL(\uni{f}) \geq |\ft{d}_0|^2 (p-1) - \normonesq{\ft{d}} \left(2 \sqrt{p} + \frac{4}{\pi} \sqrt{p} \log\left(\frac{4 p}{\pi}\right)\right)-2.
\]
\item\label{Orestes} If $d$ and $d'$ are orthogonal to each other and at least one of them is balanced, then
\[
\PCC(f,g) \leq \normone{\ft{d}}\normone{\ft{d'}}  \left(2 \sqrt{p} + \frac{4}{\pi} \sqrt{p} \log\left(\frac{4 p}{\pi}\right)\right),
\]
and if $d$ and $d'$ are also both unimodular, then
\[
\PCC(\uni{f},\uni{g}) \leq \normone{\ft{d}}\normone{\ft{d'}}  \left(2 \sqrt{p} + \frac{4}{\pi} \sqrt{p} \log\left(\frac{4 p}{\pi}\right)\right)+2.
\]
\item\label{Nestor} If neither $d$ nor $d'$ is balanced, then $\ft{d}_0\ft{d'}_0\not=0$, and
\[
\PCC(f,g) \geq |\ft{d}_0 \ft{d'}_0| p - \normone{\ft{d}}\normone{\ft{d'}} \left(2 \sqrt{p} + \frac{4}{\pi} \sqrt{p} \log\left(\frac{4 p}{\pi}\right)\right),
\]
and if $d$ and $d'$ are also both unimodular, then
\[
\PCC(\uni{f},\uni{g}) \geq |\ft{d}_0 \ft{d'}_0| p - \normone{\ft{d}}\normone{\ft{d'}} \left(2 \sqrt{p} + \frac{4}{\pi} \sqrt{p} \log\left(\frac{4 p}{\pi}\right)\right)-2,
\]
\item\label{Cyprian} If $d$ and $d'$ are not orthogonal to each other, then $\ip{\ft{d}}{\ft{d'}}\not=0$ and
\[
\PCC(f,g) \geq \ip{\ft{d}}{\ft{d'}}  \frac{p}{2} - \normone{\ft{d}}\normone{\ft{d'}} \left(2 \sqrt{p} + \frac{4}{\pi} \sqrt{p} \log\left(\frac{4 p}{\pi}\right)\right),
\]
and if $d$ and $d'$ are also both unimodular, then
\[
\PCC(\uni{f},\uni{g}) \geq \ip{\ft{d}}{\ft{d'}} \frac{p}{2} - \normone{\ft{d}}\normone{\ft{d'}} \left(2 \sqrt{p} + \frac{4}{\pi} \sqrt{p} \log\left(\frac{4 p}{\pi}\right)\right)-2.
\]
\end{enumerate}
\end{lemma}
\begin{proof}
All correlations between $p$-instances (or unimodularized $p$-instances) of cyclotomic patterns at shifts $s$ with $|s| \geq p$ vanish because the instances are supported on $\{0,1,\ldots,p-1\}$ so the $(j+s)$th term of one instance and the $j$th term of the other cannot simultaneously be nonvanishing.
Thus, the results follow from \cref{Edna} and \cref{Mildred}, where we set $e=\ft{d}$ and $e'=\ft{d'}$, use \cref{Greta} to see that $d$ (resp., $d'$) is balanced if and only if $\ft{d}_0=0$ (resp., $\ft{d'}_0=0$), and use the isometry property of the Fourier transform to see that $d$ and $d'$ are orthogonal if and only if $\ip{\ft{d}}{\ft{d'}}=0$.  
Then, in particular,
\begin{itemize}
\item the bounds on $\PSL(f)$ follow from \cref{Mildred}, where we use $|\AC_{f,f}(1)|$ as the lower bound when $d$ is not balanced;
\item the bounds on $\PSL(\uni{f})$ then follow from those on $\PSL(f)$ using \cref{Ulrich};
\item the bounds on $\PCC(f,g)$ follow from \cref{Edna}, where we use $|\AC_{f,g}(0)|$ as the lower bound when neither $d$ nor $d'$ is balanced, but when $d$ and $d'$ are not orthogonal, we use $|\AC_{f,g}(s)|$ for a shift $s$ in $\{s \in \Z: |s| \leq p/2\}$ that satisfies the congruence $s\equiv r'-r$, where $r$ (resp., $r'$) is the advancement used when rotating to obtain $f$ (resp., $g$) from the periodic sequence of length $p$ from $d$ (resp., $d'$); and
\item the bounds on $\PCC(\uni{f},\uni{g})$ then follow from those on $\PCC(f,g)$ using \cref{Ulrich}.  \qedhere
\end{itemize}
\end{proof}
We can apply this lemma to bound the correlations in codebooks from cyclotomic plans.
\begin{corollary}\label{Evelyn}
Let $n \in \Z^+$, let $D$ be a cyclotomic plan of index $n$, and let $p$ be a prime with $p\equiv 1\pmod{n}$.
Let $F$ be a $p$-instance of $D$ and, if $D$ is unimodular, let $\uni{F}$ be a unimodularized $p$-instance of $D$.
\begin{enumerate}[(i).]
\item\label{Theresa} If $D$ is balanced and orthogonal, then
\[
\GUC(F) \leq \max_{d \in D} \normonesq{\ft{d}} \left(2 \sqrt{p} + \frac{4}{\pi} \sqrt{p} \log\left(\frac{4 p}{\pi}\right)\right),
\]
and if $D$ is also unimodular, then
\[
\GUC(\uni{F}) \leq \max_{d \in D} \normonesq{\ft{d}} \left(2 \sqrt{p} + \frac{4}{\pi} \sqrt{p} \log\left(\frac{4 p}{\pi}\right)\right)+2.
\]
\item\label{Bertrand} If $D$ is not balanced, then $\ft{d}_0\not=0$ for some $d \in D$, and we have
\[
\GUC(F) \geq |\ft{d}_0|^2 (p-1) - \normonesq{\ft{d}} \left(2 \sqrt{p_\iota} + \frac{4}{\pi} \sqrt{p_\iota} \log\left(\frac{4 p_\iota}{\pi}\right)\right),
\]
and if $D$ is also unimodular, then
\[
\GUC(\uni{F}) \geq |\ft{d}_0|^2 (p-1) - \normonesq{\ft{d}} \left(2 \sqrt{p_\iota} + \frac{4}{\pi} \sqrt{p_\iota} \log\left(\frac{4 p_\iota}{\pi}\right)\right)-2.
\]
\item\label{Ophelia} If $D$ is not orthogonal, then $\ip{\ft{d}}{\ft{d'}}\not=0$ for some pair of distinct $d,d' \in D$, and if $|F|=|D|$, then we have
\[
\GUC(F) \geq \ip{\ft{d}}{\ft{d'}} \frac{p}{2} - \normone{\ft{d}}\normone{\ft{d'}} \left(2 \sqrt{p} + \frac{4}{\pi} \sqrt{p} \log\left(\frac{4 p}{\pi}\right)\right),
\]
and if $D$ is also unimodular with $|\uni{F}|=|D|$, then
\[
\GUC(\uni{F}) \geq \ip{\ft{d}}{\ft{d'}} \frac{p}{2} - \normone{\ft{d}}\normone{\ft{d'}} \left(2 \sqrt{p} + \frac{4}{\pi} \sqrt{p} \log\left(\frac{4 p}{\pi}\right)\right)-2.
\]
\end{enumerate}
\end{corollary}
\begin{proof}
Part \eqref{Theresa} follows from \cref{Martha}\eqref{Barbara} and \eqref{Orestes}.
Part \eqref{Bertrand} follows from \cref{Martha}\eqref{Ursula}.
Part \eqref{Ophelia} imposes the condition $|F|=|D|$ (resp., $|\uni{F}|=|D|$), because if $d$ and $d'$ are distinct patterns in $D$, then it is possible that their $p$-instances (resp., unimodularized $p$-instances) in $F$ (resp., $\uni{F}$) coincide, and then we should not be including their crosscorrelations when bounding $\GUC(F)$ (resp., $\GUC(\uni{F})$) from below.
Once this condition is imposed to prevent such coincidences, then the results follow from \cref{Martha}\eqref{Cyprian}.
\end{proof}
Part \eqref{Ophelia} of \cref{Evelyn} is complicated by the possibility that distinct cyclotomic patterns can produce the same $p$-instances (or unimodularized $p$-instances) for a prime $p$.
(For example, $d=(1,0)$ and $d'=(0,1)$ produce the length $p=3$ periodic sequences $f=(0,1,0)$ and $g=(0,0,1)$, respectively, and $f^{(1)}=g^{(2)}$.)
Fortunately, we can use \cref{Edna} show that $p$-instances and unimodularized $p$-instances from different cyclotomic patterns must be different if $p$ is sufficiently large.
\begin{lemma}\label{Egbert}
Let $n\in\Z^+$, let $d,d'$ be distinct cyclotomic patterns of index $n$, and for each prime $p$ with $p\equiv 1 \pmod{n}$, let $S_p$ (resp., $S'_p$) be the set of all $p$-instances of $d$ (resp., $d'$) and also all unimodularized $p$-instances of $d$ (resp., $d'$) if $d$ (resp., $d'$) is unimodular.
Then there is some $N$ such that for every prime $p$ with $p\equiv 1 \pmod{n}$ and $p \geq N$, the set $S_p\cap S'_p$ is empty.
\end{lemma}
\begin{proof}
First suppose that there is no permutation $\pi$ of $\Zn$ such that $d'=d\circ \pi$.
This means that there is some nonzero $a \in \C$ such that the number of $j \in \Zn$ with $d_j=a$ is less than the number of $j \in \Zn$ with $d'_j=a$.
Thus, a $p$-instance or unimodularized $p$-instance of $d'$ has at least $(p-1)/n-1$ more terms equal to $a$ than a $p$-instance or unimodularized $p$-instance of $d$, and so setting $N=n+2$ gives us what we seek.

Now suppose that there is a permutation $\pi$ of $\Zn$ such that $d'=d\circ \pi$.
This means that $\normtwo{d'}=\normtwo{d}$ (hence $\normtwo{\ft{d}}=\normtwo{\ft{d'}}$) and $\ft{d'}_0=\ft{d}_0$.
Since $d\not=d'$, this also means that neither $d$ nor $d'$ is of the form $(a,a,\ldots,a)$, and so neither $\ft{d}$ nor $\ft{d'}$ is of the form $(b,0,0,\ldots,0)$, and thus $|\ft{d'}_0|=|\ft{d}_0| < \normtwo{\ft{d}}=\normtwo{\ft{d'}}$.
Thus, there is some $N$ such that for every $p\geq N$, we have
\begin{equation}\label{Sally}
\frac{2+\normtwosq{\ft{d}}}{p}  + \normone{\ft{d}} \normone{\ft{d'}} \left(\frac{2}{\sqrt{p}} + \frac{4}{\pi \sqrt{p}} \log\left(\frac{4 p}{\pi}\right)\right)  < \normtwosq{\ft{d}}-|\ft{d}_0|^2.
\end{equation}
Let $p$ be a prime with $p\equiv 1 \pmod{n}$ and $p\geq N$, and suppose that $S_p \cap S'_p\not=\emptyset$ to show a contradiction.
Let $f \in S_p\cap S'_p$.
Let $g$ and $h$ be periodic sequences of length $p$ from $d$ and $d'$, respectively, and let $r, r' \in \Z$ be such that $f$ is the rotation by $r$ of either $g$ or a unimodularization of $g$, and such that $f$ is also the rotation by $r'$ of $h$ or a unimodularization of $h$.
If $r\equiv r'\pmod{p}$, then we would have $g_{j+r}=f_j=h_{j+r'}=h_{j+r}$ for all $j\in\{0,1,\ldots,p-1\}$ with $j\not\equiv -r \pmod{p}$, i.e., $g_k=h_k$ for all $k\in\Fpu$, which is impossible since $g$ and $h$ are derived from different cyclotomic patterns.
So we must have $r\not\equiv r' \pmod{p}$.
We have
\[
\AC_{f,f}(0)-\AC_{g^{(r)},h^{(r')}}(0) = \sum_{j=0}^{p-1} \left(|f_j|^2 - g_{j+r} \conj{h_{j+r'}}\right),
\]
and $f_j=g_{j+r}$ (resp., $f_j=h_{j+r'}$) for all $j\in\{0,1,\ldots,p-1\}$ with the possible exception of the one value with $j\equiv -r \pmod{p}$ (resp., $j \equiv -r' \pmod{p}$).
If we let $K$ be the set of integers in $\{0,1,\ldots,p-1\}$ such that $j\equiv -r \pmod{p}$ or $j\equiv -r' \pmod{p}$, then $|K|=2$ and
\[
\AC_{f,f}(0)-\AC_{g^{(r)},h^{(r')}}(0) = \sum_{j \in K} \left(|f_j|^2 - g_{j+r} \conj{h_{j+r'}}\right) = \sum_{j \in K} |f_j|^2,
\]
where the last equality uses the fact that $g$ and $h$ are periodic sequences derived from cyclotomic patterns.
When $j\equiv -r \pmod{p}$ (resp. $j\equiv -r' \pmod{p}$), then either $f_j=g_{j+r}=g_0=0$ (resp., $f_j=h_{j+r'}=0$) or else $|f_j|=1$ if $f$ was obtained by rotating a unimodularization of $g$ (resp., $h$).
So $|f_j|^2 \leq 1$ for all $j \in K$, and so $|\AC_{f,f}(0)-\AC_{g^{(r)},h^{(r')}}(0)| \leq |K|=2$.
\cref{Edna} (and the facts that $r\not\equiv r'\pmod{p}$ and $\ft{d}_0=\ft{d'}_0$) tells us that
\[
|\AC_{g^{(r)},h^{(r')}}(0)| \leq p |\ft{d}_0|^2 + \normone{\ft{d}} \normone{\ft{d'}} \left(2 \sqrt{p} + \frac{4}{\pi} \sqrt{p} \log\left(\frac{4 p}{\pi}\right)\right),
\]
and so
\begin{equation}\label{Desdemona}
|\AC_{f,f}(0)| \leq 2  + p |\ft{d}_0|^2 + \normone{\ft{d}} \normone{\ft{d'}} \left(2 \sqrt{p} + \frac{4}{\pi} \sqrt{p} \log\left(\frac{4 p}{\pi}\right)\right) .
\end{equation}
If $f$ is a $p$-instance of $d$ (resp., $d'$), then \cref{Abigail} shows that $\AC_{f,f}(0)=\normtwosq{f}=(p-1) \normtwosq{\ft{d}}$ (resp., $\AC_{f,f}(0)=\normtwosq{f}=(p-1) \normtwosq{\ft{d'}}$), but if $f$ is a unimodularized $p$-instance of $d$ (resp., $d'$), then $\AC_{f,f}(0)=1+(p-1) \normtwosq{\ft{d}}$ (resp., $\AC_{f,f}(0)=1+(p-1) \normtwosq{\ft{d'}}$), and since $\normtwosq{\ft{d}}=\normtwosq{\ft{d'}}$ we have, in any case, $|\AC_{f,f}(0)| \geq (p-1) \normtwosq{\ft{d}}$.
Combining this inequality transitively with \eqref{Desdemona} and rearranging, we see that we contradict \eqref{Sally}, and so in fact $S_p\cap S'_p=\emptyset$.
\end{proof}
The consequence of this last lemma is that, for sufficiently large primes, instances of cyclotomic plans contain one sequence for each pattern in the plan.
\begin{corollary}\label{Nora}
Let $D$ be a cyclotomic plan of index $n$.
There is some $N$ such that for every prime $p$ with $p\equiv 1 \pmod{n}$ and $p \geq N$, every $p$-instance of $D$ (and if $D$ is unimodular, every unimodularized $p$-instance of $D$) contains precisely $|D|$ sequences.
\end{corollary}
\begin{proof}
For each pair $d,d'$ of distinct cyclotomic patterns in $D$, \cref{Egbert} provides a number $N_{d,d'}$ such that whenever $p\geq N_{d,d'}$, no $p$-instance or unimodularized $p$-instance of $d$ can equal any $p$-instance or unimodularized $p$-instance of $d'$, so we set $N$ to be the maximum $N_{d,d'}$ over all distinct pairs $d,d' \in D$.
\end{proof}

\subsection{Balance and orthogonality regulate growth of $\GUC$}\label{Margaret}

We prove \cref{Penelope} in this section.
We first prove that balance and orthogonality of a cyclotomic plan implies well regulation of the growth of $\GUC$ for instances of that plan.
\begin{proposition}\label{Winston}
Let $D$ be a cyclotomic plan with $0\not\in D$.
Let $\calA$ be an infinite family of instances of $D$.
Then $\calA$ has well regulated growth of $\GUC$ if and only if $D$ is both balanced and orthogonal.
\end{proposition}
\begin{proof}
Write $\calA=\{A_\iota\}_{\iota \in I}$.
For each $\iota \in I$, let $p_\iota$ denote a prime such that $A_\iota$ is an $p_\iota$-instance of $D$.
Note that $\{p_\iota: \iota \in I\}$ is unbounded since there are only finitely many distinct $p$-instances of $D$ for any given $p$ (since there at most $p$ distinct rotations of any given periodic sequence of length $p$).
\cref{Abigail} shows that
\begin{equation}\label{Justin}
\SDC(A_\iota)=\frac{p_\iota-1}{|D|} \min_{d \in D} \normtwosq{d},
\end{equation}
and since $\min_{d \in D} \normtwosq{d}$ is positive (because $0\not\in D$), taking the limit as $\SDC(A_\iota)\to\infty$ is the same as taking the limit as $p_\iota\to\infty$.

\cref{Nora} furnishes an $N$ such that whenever $p\geq N$, a $p$-instance of $D$ must contain precisely $|D|$ sequences.
We let $J=\{\iota \in I: p_\iota \geq N\}$, and then it suffices to prove that the family $\{A_\iota\}_{\iota \in J}$ has well regulated growth of $\GUC$, because \eqref{Justin} shows that $\{\SDC(A_\iota): \iota \in I\smallsetminus J\}$ is a bounded set (with upper bound $(N-2) \min_{d \in D} \normtwosq{d}/|D|$), so restricting from $I$ to $J$ does not have any effect when we take the limit as $\SDC(A_\iota)$ tends to infinity.
Restricting from $I$ to $J$ allows us to apply all three parts of \cref{Evelyn}.

If $D$ is not balanced, then \cref{Evelyn}\eqref{Bertrand} says that there is some $d\in D$ with $\ft{d}_0\not=0$, and for each $\iota \in J$,
\[
\GUC(A_\iota) \geq |\ft{d}_0|^2 (p_\iota-1)- \normonesq{\ft{d}} \left(2 \sqrt{p_\iota} + \frac{4}{\pi} \sqrt{p_\iota} \log\left(\frac{4 p_\iota}{\pi}\right)\right).
\]
Since $|\ft{d}_0|^2$ and $\min_{d \in D} \normtwosq{d}$ are positive, this inequality and \eqref{Justin} show that $\GUC(A_\iota)/\SDC(A_\iota)^{3/4}$ tends to infinity as $\SDC(A_\iota)\to\infty$, so that $\{A_\iota\}_{\iota \in J}$ does not have well regulated growth of $\GUC$.

If $D$ is not orthogonal, then by \cref{Evelyn}\eqref{Ophelia} there are distinct $d,d' \in D$ with $\ip{\ft{d}}{\ft{d'}}\not=0$, and for each $\iota \in J$,
\begin{align*}
\GUC(A_\iota) \geq \ip{\ft{d}}{\ft{d'}} \frac{p_\iota}{2} - \normone{\ft{d}}\normone{\ft{d'}} \left(2 \sqrt{p_\iota} + \frac{4}{\pi} \sqrt{p_\iota} \log\left(\frac{4 p_\iota}{\pi}\right)\right).
\end{align*}
Since $\ip{\ft{d}}{\ft{d'}}\not=0$ and $\min_{d \in D} \normtwosq{d}$ are positive, this inequality and \eqref{Justin} show that $\GUC(A_\iota)/\SDC(A_\iota)^{3/4}$ tends to infinity as $\SDC(A_\iota)\to\infty$, so that $\{A_\iota\}_{\iota \in J}$ does not have well regulated growth of $\GUC$.

If $D$ is balanced and orthogonal, then for $\iota \in J$, \cref{Evelyn}\eqref{Theresa} shows that
\[
\GUC(A_\iota) \leq \max_{d \in D} \normonesq{\ft{d}} \left(2 \sqrt{p_\iota} + \frac{4}{\pi} \sqrt{p_\iota} \log\left(\frac{4 p_\iota}{\pi}\right)\right).
\]
Since $\min_{d \in D} \normtwosq{d}$ is positive, this inequality and \eqref{Justin} show that if $\epsilon>0$, then $\GUC(A_\iota)/\SDC(A_\iota)^{1/2+\epsilon}$ tends to zero as $\SDC(A_\iota)\to\infty$, so that $\{A_\iota\}_{\iota \in J}$ has well regulated growth of $\GUC$.
\end{proof}
We now restate and prove our first main result, \cref{Penelope}, which is the analogue of \cref{Winston} for unimodular instances of a cyclotomic plan.
\begin{theorem}\label{Samuel}
Let $D$ be a unimodular cyclotomic plan.
Let $\calA$ be a family of unimodularized instances of $D$ such that the $\{\len(A): A \in \calA\}$ is unbounded.
Then $\calA$ has well regulated growth of $\GUC$ if and only if $D$ is both balanced and orthogonal.
\end{theorem}
\begin{proof}
Write $\calA=\{A_\iota\}_{\iota \in I}$.
For each $\iota\in I$, let $p_\iota=\len(A_\iota)$, let $P_\iota$ be the periodic codebook of length $p_\iota$ derived from $D$, and for each $f \in P_\iota$, let $\uni{f}$ be a unimodularization of $f$ and let $r_f$ be an integer such that $\{\uni{f}^{(r_f)}: f \in P_\iota\}=A_\iota$, and then set $\uni{P}_\iota=\{\uni{f}: f \in P_\iota\}$ and $B_\iota=\{f^{(r_p)}: f \in P_\iota\}$.
By \cref{Nora}, there is some $N$ such that whenever $p_\iota \geq N$, we have $|A_\iota|=|B_\iota|=|D|$, and we set $J=\{\iota \in I: p_\iota \geq N\}$.
Since well regulated growth of $\GUC$ for $\{A_\iota\}_{\iota \in I}$ concerns the limit as $\SDC(A_\iota)=\len(A_\iota)=p_\iota$ tends to infinity, it suffices for us to prove that $\{A_\iota\}_{\iota \in J}$ has well regulated growth of $\GUC$.
Then $\{B_\iota\}_{\iota \in J}$ must be an infinite family of instances of $D$ because $\{\len(A_\iota): \iota \in J\}=\{p_\iota: \iota \in J\}$ is unbounded.
The unimodularity of $D$ implies that $0 \not\in D$, so we may apply \cref{Winston} to see that $\{B_\iota\}_{\iota \in J}$ has well regulated growth of $\GUC$ if and only if $D$ is both balanced and orthogonal.
And \cref{Melanie} (with $\{B_\iota\}_{\iota \in J}$ and $\{A_\iota\}_{\iota \in J}$ here taking the respective roles of $\{A_\iota\}_{\iota\in I}$ and $\{\uni{A}_\iota\}_{\iota\in I}$ in the corollary) shows that $\{B_\iota\}_{\iota \in J}$ has well regulated growth of $\GUC$ if and only if $\{A_\iota\}_{\iota \in J}$ does.
\end{proof}

\section{Demerit factor of codebooks from cyclotomic plans}\label{Genevieve}

This section is dedicated to proving \cref{Mary}.
\cref{Irving} introduces periodic correlation, which has a profound connection with aperiodic correlation, which will become manifest in results like \cref{Adam} and \cref{Aaron}.
In \cref{Katherine} we use cyclotomic numbers to compute periodic correlation for sequences from cyclotomic patterns.
These results are used in \cref{Jane} to approximate aperiodic demerit factors for instances of cyclotomic patterns.
In \cref{Madeleine} we specifically consider aperiodic demerit factors for instances of Hadamard plans, and prove \cref{Mary} as a corollary of \cref{Victor}, which allows for more freedom than \cref{Mary} as to how one rotates sequences, but at the expense of a restriction on the primes used to make instances.

\subsection{Periodic correlation}\label{Irving}

There is a periodic version of correlation, which we shall use to elucidate facts about the aperiodic correlation.
If $\ell \in \Z^+$, $f$ and $g$ are periodic sequences of length $\ell$, and $s \in \Zell$, then the {\it periodic crosscorrelation of $f$ with $g$ at shift $s$} is defined to be
\[
\PC_{f,g}(s) = \sum_{j \in \Zell} f_{j+s} \conj{g_j}.
\]
If $s \in \Z$, we use the convention that $\PC_{f,g}(s)$ means $\PC_{f,g}(s+\ell\Z)$; here, $s+\ell\Z$ is the reduction of $s$ modulo $\ell$.
Observe that the periodic crosscorrelation of two periodic sequences of length $\ell$ is deducible the aperiodic crosscorrelation: for every $r, s \in \Z$ with $0 \leq s < \ell$, we have
\[
\PC_{f,g}(s)= \AC_{f^{(r)},g^{(r)}}(s) + \AC_{f^{(r)},g^{(r)}}(s-\ell).
\]
In particular, note that $\PC_{f,g}(0)=\AC_{f^{(r)},g^{(r)}}(0)$ since $\AC_{f^{(r)},g^{(r)}}(-\ell)=0$ because $f^{(r)}_{j-\ell}$ and $g^{(r)}_j$ cannot simultaneously have their indices in the set $\{0,1,\ldots,\ell-1\}$ upon which $f^{(r)}$ and $g^{(r)}$ are supported.
The periodic correlation of sequences at shift $0$ is the same as their inner product, and indeed $\PC_{f,g}(0)=\ip{f}{g}=\ip{f^{(r)}}{g^{(r)}}=\AC_{f^{(r)},g^{(r)}}(0)$ for every $r \in \Z$.
The {\it periodic autocorrelation of $f$ at shift $s$} is the periodic crosscorrelation of $f$ with itself at shift $s$, that is, $\PC_{f,f}(s)$, and note that $\PC_{f,f}(0)=\normtwosq{f}=\normtwosq{f^{(r)}}=\AC_{f^{(r)},f^{(r)}}(0)$ for every $r \in \Z$.

We have a {\it periodic crosscorrelation demerit factor of $f$ with $g$} for two nonzero periodic sequences, $f$ and $g$, of the same length $\ell$,
\[
\PCDF(f,g)=\frac{\sum_{s \in \Zell} |\PC_{f,g}(s)|^2}{\PC_{f,f}(0) \PC_{g,g}(0)},
\]
and a {\it periodic autocorrelation demerit factor of $f$} for $f$ a nonzero periodic sequence of length $\ell$,
\[
\PADF(f)=\frac{\sums{s \in \Zell \\ s\not=0} |\PC_{f,f}(s)|^2}{\PC_{f,f}(0)^2} = \PCDF(f,f)-1,
\]
and one can define periodic merit factors as the reciprocals of periodic demerit factors (when they are nonzero).
For a periodic codebook $F$, we define the {\it periodic (crosscorrelation) demerit factor of $F$} to be
\[
\PCDF(F)=\frac{1}{|F|^2} \sum_{f,g \in F} \PCDF(f,g).
\]

We prove an upper bound on the periodic crosscorrelation demerit factor of a sequence.
\begin{lemma}\label{Raymond}
If $f$ and $g$ are periodic sequences of length $\ell$, then we have  $\PCDF(f,g) \leq \ell$.
\end{lemma}
\begin{proof}
For each $s \in \Zell$, the correlation $\PC_{f,g}(s)$ is the inner product of $f$ and some sequence $h$ which is a cyclically shifted version of $g$.
Thus, by the Cauchy--Schwarz inequality, we have $|\PC_{f,g}(s)| \leq \normtwo{f} \normtwo{h}=\normtwo{f}\normtwo{g}$.
Thus, $\sum_{s \in \Zell} |\PC_{f,g}(s)|^2 \leq \ell \normtwosq{f} \normtwosq{g}$, and since $\PC_{f,f}(0)=\normtwosq{f}$ and $\PC_{g,g}(0)=\normtwosq{g}$, this means that $\PCDF(f,g) \leq \ell$.
\end{proof}

We identify periodic sequences of length $\ell$ with elements of the quotient ring $\quotring{\ell}$, where $z$ is an indeterminate and $(z^\ell-1)$ is the ideal in $\C[z]$ generated by $z^\ell-1$, and we let $y=z+(z^\ell-1)$.
Since $y$ is an element of order $\ell$ in $\quotring{\ell}$, we may write $y^j$ when $j \in \Zell$ without ambiguity, and every element $f \in \quotring{\ell}$ can be written uniquely as $f=\sum_{j \in \Zell} f_j y^j$, where each $f_j \in \C$; we call this the {\it canonical representation of $f$} and we identify it with the periodic sequence $j\mapsto f_j$ (i.e., $(f_0,f_1,\ldots,f_{\ell-1})$).
Thus, if $h\in\quotring{\ell}$ and $j\in\Zell$, the notation $h_j$ means the coefficient of $y^j$ in the canonical representation of $h$.
If $f=\sum_{j \in \Zell} f_j y^j \in \quotring{\ell}$, then we set $\conj{f}=\sum_{j \in \Zell} \conj{f_j} y^{-j}$, and note that $f \mapsto \conj{f}$ is an involutional automorphism of our ring $\quotring{\ell}$.
Then it is not difficult to show that if $f,g \in \quotring{\ell}$ represent periodic sequences of length $\ell$, then we have
\begin{equation}\label{Isidore}
f \conj{g} = \sum_{s \in \Zell} \PC_{f,g}(s) y^s,
\end{equation}
and in particular, this means that $\ip{f}{g}=(f\conj{g})_0$, and so $\PC_{f,f}(0)=\normtwosq{f}=(f\conj{f})_0$.
Thus,
\begin{equation}\label{Hubert}
\sum_{s \in \Zell} |\PC_{f,g}(s)|^2 = \normtwosq{f\conj{g}} = (f\conj{f}g\conj{g})_0=\normtwosq{f g},
\end{equation}
and so if $f,g\not=0$, we have
\[
\PCDF(f,g)=\frac{\normtwosq{f g}}{\normtwosq{f} \normtwosq{g}},
\]
and
\[
\PADF(f)=\frac{\normtwosq{f^2}}{\normtwo{f}^4}-1.
\]

If $F$ is a uniform-length periodic codebook in which every sequence has the same $l^2$ norm, then one can calculate its periodic crosscorrelation demerit factor entirely from autocorrelation values.
\begin{lemma}\label{Bruno}
Let $F$ be a periodic codebook of length $\ell$.
Then
\[
\sum_{f,g \in F} \sum_{s \in \Zell} |\PC_{f,g}(s)|^2 = \sum_{s \in \Zell} \left|\sum_{f \in F} \PC_{f,f}(s) \right|^2,
\]
and furthermore, if there is some nonzero $C$ such that $\PC_{f,f}(0)=C$ for every $f \in F$, then
\[
\PCDF(F)=1+\frac{1}{|F|^2 C^2} \sums{s \in \Zell\\ s\not=0} \left|\sum_{f \in F} \PC_{f,f}(s) \right|^2.
\]
\end{lemma}
\begin{proof}
We have
\begin{align*}
\sum_{f,g \in F} \sum_{s \in \Zell} |\PC_{f,g}(s)|^2
& = \sum_{f,g \in F} (f\conj{f} g\conj{g})_0 && \text{from \eqref{Hubert}}\\
& = \left(\left(\sum_{f \in F} f\conj{f}\right) \conj{\left(\sum_{g \in F} g\conj{g}\right)}\right)_0 \\
& = \Normtwosq{\sum_{f \in F} f\conj{f}} \\
& = \Normtwosq{\sum_{f \in F} \sum_{s \in \Zell} \PC_{f,f}(s) y^s} && \text{from \eqref{Isidore}} \\
& = \sum_{s \in \Zell} \left|\sum_{f \in F} \PC_{f,f}(s)\right|^2,
\end{align*}
which proves the first claim.
If there is some $C\not=0$ such that $\PC_{f,f}(0)=C$ for every $f \in F$, then divide this equation through by $C^2 |F|^2$ and separate out the $s=0$ term from the sum to obtain the second claim.
\end{proof}
A {\it periodically complementary} codebook is uniform-length periodic codebook $F$ such that $\sum_{f \in F} \PC_{f,f}(s)=0$ for all nonzero $s$.
Our last result shows that among the uniform-length periodic codebooks in which all sequences have the same $l^2$ norm (these include all the uniform-length, unimodular periodic codebooks), the periodically complementary ones are the ones with smallest $\PCDF$.
\begin{corollary}\label{Milton}
Let $F$ be a uniform-length periodic codebook in which all the sequences have the same periodic autocorrelation at shift $0$.
Then we have $\PCDF(F)\geq 1$, with equality if and only if $F$ is periodically complementary.
\end{corollary}

\subsection{Periodic correlation of sequences from cyclotomic patterns}\label{Katherine}

To compute periodic correlations of sequences from cyclotomic patterns, we use cyclotomic numbers.
For $p$ a prime, $n$ a positive integer with $n\mid p-1$, and $j,k \in \Zn$, we define the {\it $j,k$ cyclotomic number over $\Fp$} to be
\[
\cyc{j}{k} = \left|(1+\alpha_p^j \Fpn) \cap (\alpha_p^k \Fpn)\right|,
\]
where we recall that $\alpha_p$ is our fixed primitive element for $\Fp$, where $\alpha_p^j\Fpn$ is shorthand for $(\alpha_p\Fpn)^j$ (which is well defined since $\alpha_p\Fpn$ has order $n$), and where we use $1+\alpha_p^j\Fpn$ to mean $\{1+b: b \in \alpha_p^j\Fpn\}$.
We use the convention that if we write a cyclotomic number $\cyc{j}{k}$ with an element of $\Z$ rather than $\Zn$ for either $j$ or $k$, then we really mean to replace that integer $u$ with $u+\Zn$.

We first give some basic identities for cyclotomic numbers that we need.
\begin{lemma}\label{Magnus}
Let $n \in \Z^+$ and let $p$ be a prime with $p\equiv 1 \pmod{n}$.
\begin{enumerate}[(i).]
\item\label{Leonard} For every $j,k \in \Zn$, we have $\cyc{j}{k}=\cyc{-j}{k-j}$.
\item\label{Muriel} For $k \in \Zn$, we have
\[
\sum_{j \in \Zn} \cyc{j}{k} = \begin{cases}
-1+\frac{p-1}{n} & \text{if $k=0$}, \\
\frac{p-1}{n} & \text{otherwise.}
\end{cases}
\]
\item\label{James} We have $\sum_{j \in \Zn} \cyc{j}{j}=-1+(p-1)/n$.
\item\label{Eric} We have $\sum_{j,k\in\Zn} \cyc{j}{k}=p-2$.
\end{enumerate}
\end{lemma}
\begin{proof}
The identity in \eqref{Leonard} is well known: see \cite[eq.(11)]{Dickson}.
Since the set of cyclotomic classes form a partition of $\Fpu$, the sum $\sum_{j \in \Zn} \cyc{j}{k}$ is the number of elements in $(1+\Fpu)\cap \alpha_p^k\Fpn = \alpha_p^k\Fpn\smallsetminus\{1\}$.
Each cyclotomic class has $(p-1)/n$ elements, and $1 \in \alpha_p^0\Fpn$, so that \eqref{Muriel} follows.
By the same principle, $\sum_{j,k \in \Zn} \cyc{j}{k}$ is  $|(1+\Fpu)\cap\Fpu|=|\Fpu\smallsetminus\{1\}|=p-2$; this establishes \eqref{Eric}.
By \eqref{Leonard}, we have $\sum_{j \in \Zn} \cyc{j}{j}=\sum_{j \in \Zn} \cyc{-j}{0}$, whose value is given by \eqref{Muriel}; this establishes \eqref{James}.
\end{proof}
Now we can calculate the periodic correlation for a pair of sequences from cyclotomic patterns.
\begin{lemma}\label{Paul}
Let $n\in\Z^+$, let $d,d'$ be cyclotomic patterns of index $n$, let $p$ be a prime with $p\equiv 1 \pmod{n}$, and let $f$ and $g$ be the periodic sequences of length $p$ from $d$ and $d'$, respectively.  Then
\[
\PC_{f,g}(0)= \frac{p-1}{n} \ip{d}{d'} =(p-1)\ip{\ft{d}}{\ft{d'}},
\]
and for $u \in \Z$, we have
\[
\PC_{f,g}(\alpha_p^u)= \sum_{j,k \in \Z/n\Z} \cyc{k}{j} d_{j+u} \conj{d'_{k+u}}.
\]
\end{lemma}
\begin{proof}
The equalities for $\PC_{f,g}(0)$ follow from \cref{Abigail} and the fact that $\PC_{f,g}(0)=\AC_{f^{(0)},g^{(0)}}(0)$.
For $u \in \Z$, we have
\[
\PC_{f,g}(\alpha_p^u) = \sum_{t \in \Fp} f_{t+\alpha_p^u} \conj{g_t}.
\]
Reparameterize with $s=\alpha_p^{-u} t$ to obtain  
\[
\PC_{f,g}(\alpha_p^u) = \sum_{s \in \Fp} f_{\alpha_p^u (s+1)} \conj{g_{\alpha_p^u s}}.
\]
If $s \in \{0,-1\}$, then $f_{\alpha_p^v (s+1)} \conj{g_{\alpha_p^v s}}=0$ because $f_0=g_0=0$; otherwise, there is a unique pair $j,k \in \{0,1,\ldots,n-1\}$ such that $s \in \alpha_p^k \Fpn$ while $s+1 \in \alpha_p^j\Fpn$, in which case $f_{\alpha_p^u (s+1)} \conj{g_{\alpha_p^v s}}=d_{j+u} \conj{d'_{k+u}}$.
Thus
\begin{align*}
\PC_{f,g}(\alpha_p^u)
& = \sum_{j=0}^{n-1} \sum_{k=0}^{n-1} |(\alpha_p^k\Fpn+1)\cap\alpha_p^j\Fpn| d_{j+u} \conj{d'_{k+u}} \\
& = \sum_{j=0}^{n-1} \sum_{k=0}^{n-1} \cyc{k}{j} d_{j+u} \conj{d'_{k+u}}. \qedhere
\end{align*}
\end{proof}
We now wish to compute the periodic crosscorrelation demerit factor of a periodic codebook derived from a Hadamard plan.
First we need a technical lemma.
\begin{lemma}\label{Oswald}
Let $D$ be a Hadamard plan of index $n$ and $j,k \in \Zn$.
Then
\[
\sum_{d \in D} d_j \conj{d_k} = \begin{cases}
n-1 & \text{if $j=k$,} \\
-1 & \text{otherwise}
\end{cases}
\]
\end{lemma}
\begin{proof}
If we think of the $n\times n$ matrix $M$ whose first row is $(1,1,\ldots,1)$ and whose remaining rows are the cyclotomic patterns in $D$, then the unimodularity, balance, and orthogonality of $D$ ensure that $M M^*=n I$ (and so $M^* M=n I$).  Now note that the sum we wish to evaluate is equal to the inner product of the $j$th and $k$th columns of the matrix obtained by deleting the first row of $M$.  So our sum is $-1$ plus the inner product of columns $j$ and $k$ of $M$.  Since $M^* M=n I$, our sum is $-1+n$ if $j=k$ and $-1+0$ otherwise.
\end{proof}
Now we are ready to compute the $\PCDF$ of a codebook from a Hadamard plan.
\begin{lemma}\label{Heloise}
Let $D$ be a Hadamard plan of index $n$, let $p$ be a prime with $p\equiv 1 \pmod{n}$, and let $F$ be the codebook of length $p$ derived from $D$.
Then for $s \in \Fp$, we have
\[
\sum_{f \in F} \PC_{f,f}(s) =
\begin{cases}
(n-1)(p-1) & \text{if $s=0$,} \\
1-n & \text{otherwise,}
\end{cases}
\]
so that
\begin{align*}
\sum_{s \in \Fpu} \left|\sum_{f \in F} \PC_{f,f}(s)\right|^2 & =(p-1)(n-1)^2, \\
\sum_{s \in \Fp} \left|\sum_{f \in F} \PC_{f,f}(s)\right|^2 & =p(p-1)(n-1)^2,
\end{align*}
and $\PCDF(F)=p/(p-1)$.
\end{lemma}
\begin{proof}
Since $D$ is unimodular, every periodic sequence $f$ in $F$ has $\PC_{f,f}(0)=p-1$ by \cref{Paul}, so that $\sum_{f \in F} \PC_{f,f}(0)=|F|(p-1)=|D|(p-1)=(n-1)(p-1)$.
If $u \in \Z$, $d \in D$, and $f$ is the sequence in $F$ from $d$, then \cref{Paul} tells us that $\PC_{f,f}(\alpha_p^u)= \sum_{j,k \in \Zn} \cyc{k}{j} d_{j+u} \conj{d_{k+u}}$, so (using the Kronecker delta notation) we have
\begin{align*}
\sum_{f \in F} \PC_{f,f}(\alpha_p^{u})
& =\sum_{d \in D} \sum_{j,k \in \Zn} \cyc{k}{j} d_{j+u} \conj{d_{k+u}} \\
& = \sum_{j,k \in \Zn} \cyc{k}{j} (-1+n\delta_{j+u,k+u}) && \text{by \cref{Oswald}} \\
& = -\sum_{j,k \in \Zn} \cyc{k}{j} +n \sum_{j \in \Z} \cyc{j}{j} \\
& = 1-n && \text{by \cref{Magnus}\eqref{Eric},\eqref{James}}.
\end{align*}
Thus, $\sum_{s \in \Fpu} \left|\sum_{f \in F} \PC_{f,f}(s)\right|^2 = (p-1) (1-n)^2$, and
\[
\sum_{s \in \Fp} \left|\sum_{f \in F} \PC_{f,f}(s)\right|^2 = ((p-1)(n-1))^2+(p-1)(n-1)^2 = p(p-1)(n-1)^2.
\]
Since $\PC_{f,f}(0)=p-1$ for every $f \in F$, \cref{Bruno} tells us that
\[
\PCDF(F)=1+\frac{(p-1)(n-1)^2}{(n-1)^2 (p-1)^2}=\frac{p}{p-1}.\qedhere
\]
\end{proof}

\subsection{Aperiodic correlation of sequences from cyclotomic patterns}\label{Jane}

Boothby and Katz \cite[Theorem 1]{Boothby-Katz} give a formula for the asymptotic crosscorrelation demerit factor of a family of pairs of sequences obtained from character patterns.
Here we present a special case of the approximation of the crosscorrelation demerit factor for a pair of sequences that they used.
To present this approximation, we need to define some parameters related to the sequence pair.
Let $n \in \Z^+$, let $d$ and $d'$ be nonzero balanced cyclotomic patterns of index $n$, let $p$ be a prime with $p\equiv 1 \pmod{n}$, let $\accentset{\circ}{f}$ and $\accentset{\circ}{g}$ be the periodic sequences of length $p$ derived from $d$ and $d'$, respectively, let $r,r' \in \Z$, and let $f=\accentset{\circ}{f}^{(r)}$ and $g=\accentset{\circ}{g}^{(r')}$.
Following eqs.~(16)--(17) and  Proposition 7 of \cite{Boothby-Katz}, we define
\begin{align*}
U_{d,d'} & = \frac{|\ip{\ft{d}}{\ft{d'}}|^2}{\normtwosq{\ft{d}} \normtwosq{\ft{d'}}}, \\ 
V_{d,d',p} & = \frac{\left|\sum_{j \in \Zn} \ft{d}_j \ft{d'}_{-j} \omega_p^{(p-1)j/n}(-1)\right|^2}{\normtwosq{\ft{d}} \normtwosq{\ft{d'}}}, \\
S_{d,d',p} & = - 1 - U_{d,d'}-V_{d,d',p} + \frac{1}{p(p-1) \normtwosq{\ft{d}} \normtwosq{\ft{d'}}} \sum_{a \in \Fp} |\PC_{\accentset{\circ}{f},\accentset{\circ}{g}}(a)|^2,  \\
W_d & = \frac{\normone{\ft{d}}}{\normtwo{\ft{d}}}\text{, and} \\
W_{d'} & = \frac{\normone{\ft{d'}}}{\normtwo{\ft{d'}}}.
\end{align*}
If $\normtwo{\ft{d}}=\normtwo{\ft{d'}}=1$, then these are the parameters $U_{f,g}$, $V_{f,g}$, $S_{f,g}$, $W_f$, and $W_g$ from \cite{Boothby-Katz}, but recast into the notation of this paper.
Boothby and Katz assume that all character patterns are normalized in their paper (see the first paragraph of their Section III), so our definitions generalize theirs to accommodate all nonzero $d$ and $d'$, and it is clear that these parameters do not change if we replace either $d$ or $d'$ with a nonzero scalar multiple of itself (since $\ft{d}$ and $\accentset{\circ}{f}$ scale as $d$ does and $\ft{d'}$ and $\accentset{\circ}{g}$ scale as $d'$ does).
The sequence that we call $\accentset{\circ}{f}$ (resp., $\accentset{\circ}{g}$) here is a linear combination of characters from $\mchars$ that in \cite{Boothby-Katz} is written as $\sum_{\phi\in\mchars} f_\phi \phi$ (resp., $\sum_{\phi\in\mchars} g_\phi \phi$); here it is written as a linear combination of characters in the subgroup $\{\omega_p^{j(p-1)/n}: 0 \leq j < n\}$ of index $n$ in $\mchars$ and is derived from character pattern $\ft{d}$ (resp., $\ft{d'}$), so that $f=\sum_{j \in \Zn} \ft{d}_j \omega_p^{(p-1) j/n}$ (resp., $g=\sum_{j \in \Zn} \ft{d'}_j \omega_p^{(p-1) j/n}$).
This correspondence enables one to derive our formulas for $U_{d,d'}$, $V_{d,d',p}$, $W_d$, and $W_{d'}$ (when $\normtwo{\ft{d}}=\normtwo{\ft{d'}}=1$) from the definitions of $U_{f,g}$, $V_{f,g}$, $W_f$, and $W_g$, respectively, in eqs.~(17) of \cite{Boothby-Katz}.
Boothby and Katz's quantity $S_{f,g}$, which corresponds to our $S_{d,d',p}$, is also defined in eq.~(16) of \cite{Boothby-Katz}, but later, in Proposition 7 of the same paper, there is a relation between $S_{f,g}$, $U_{f,g}$, $V_{f,g}$ and the periodic correlations of $\accentset{\circ}{f}$ and $\accentset{\circ}{g}$ (which are called $\per(f)$ and $\per(g)$ by Boothby and Katz), and it is from this that we obtain an equivalent definition of $S_{f,g}$, which we translate into our notation as $S_{d,d',p}$.

The parameters $U_{d,d'}$, $V_{d,d',p}$, and $S_{d,d',p}$ are defined in terms of $\ft{d}$ and $\ft{d'}$, but it will be helpful to have formulae for them in terms of $d$ and $d'$.
If $d$ and $d'$ are cyclotomic plans of index $n$, and $\sigma \in \Z/2\Z$, we define
\[
V^{(\sigma)}_{d,d'} =\frac{1}{\normtwosq{d}\normtwosq{d'}} \left|\sum_{j \in \Zn} d_j d'_{j+\sigma(n/2)}\right|^2,
\]
where the translation $\sigma(n/2)$ in the index of $d'$ is well defined since the index is read modulo $n$.
(If $n$ is even, then the value of $h(n/2)$ for $h \in \Z$ only depends on the parity of $h$, and if $n$ is odd, then we must consider the reduction modulo $n$ of the half-integer $n/2$ to be the element $n 2^{-1}=0$ in $\Z/n\Z$, so $h(n/2)\equiv 0 \pmod{n}$ for any $h \in \Z$, and so $V^{(0)}_{d,d'}=V^{(1)}_{d,d'}$ whenever $n$ is odd.)
If $k$ is in $\Z$, then we take $V^{(k)}_{d,d'}$ to mean $V^{(\sigma)}_{d,d'}$ where $\sigma=k+2\Z \in \Z/2\Z$.
\begin{lemma}\label{Matilda}
Let $n \in \Z^+$, let $d$ and $d'$ be balanced nonzero cyclotomic patterns of index $n$, let $p$ be a prime with $p\equiv 1 \pmod{n}$, and let $\accentset{\circ}{f}$ (resp., $\accentset{\circ}{g}$) be the periodic sequence of length $p$ from $d$ (resp., $d'$).
\begin{enumerate}[(i).]
\item\label{John} We have \[U_{d,d'} = \frac{\left|\ip{d}{d'}\right|^2}{\normtwosq{d}\normtwosq{d'}},\] so $U_{d,d'}$ is a real number in $[0,1]$, and equals $1$ when $d=d'$.
\item\label{Gordon} We have $V_{d,d',p} = V^{((p-1)/n)}_{d,d'}$, so that $V_{d,d',p}$ is a real number in $[0,1]$.
\item\label{Annette} We have $S_{d,d',p}=\frac{p-1}{p} \PCDF(\accentset{\circ}{f},\accentset{\circ}{g}) - 1 - U_{d,d'} - V_{d,d',p}$, so $S_{d,d',p}$ is a real number.
\end{enumerate}
\end{lemma}
\begin{proof}
The formula for $U_{d,d'}$ follows from the isometric property of the Fourier transform, and this clearly gives a nonnegative real number.
Setting $d=d'$ gives $U_{d,d'}=1$, and otherwise one uses the Cauchy--Schwarz inequality to upper bound the inner product $\ip{d}{d'}$ by $\normtwo{d}\normtwo{d'}$, so $U_{d,d'} \leq 1$.

Now we prove that $V_{d,d',p}=V^{((p-1)/n)}_{d,d'}$.
If $n=1$, the result is clear because then our patterns are of length $1$, and so $d=\ft{d}$, $d'=\ft{d'}$, $\normtwo{d}=\normtwo{\ft{d}}=|d_0|$, $\normtwo{d'}=\normtwo{\ft{d'}}=|d'_0|$, and so $V_{d,d',p}=1=V^{(0)}_{d,d'}=V^{(1)}_{d,d'}$.  So henceforth assume that $n>1$, which forces $p$ to be odd.
Then $V_{d,d',p}=|A|^2$, where
\begin{align*}
A
& =\frac{1}{\normtwo{\ft{d}}\normtwo{\ft{d'}}} \sum_{h \in \Zn} \ft{d}_h \ft{d'}_{-h} \omega_p^{(p-1)h /n}(-1) \\
& =\frac{n}{\normtwo{d} \normtwo{d'}} \sum_{h \in \Zn} \ft{d}_h \ft{d'}_{-h} \omega_p^{(p-1)h /n}(-1).
\end{align*}
Using the definitions of the Fourier transform and $\omega_p$ along with the fact that $-1=\alpha_p^{(p-1)/2}$ (so that $\omega_p(-1)=-1$), we obtain the following formulae (which use the Kronecker delta):
\begin{align*}
A
& = \frac{\sum_{h,j,k \in \Zn} d_j \exp(-2\pi i h j/n) d'_k \exp(2\pi i h k/n)\exp((p-1) \pi i h/n)}{n \normtwo{d}\normtwo{d'}}  \\
& = \frac{\sum_{j,k \in \Zn} d_j d'_k \sum_{h \in \Zn} \exp(2\pi i h (k-j+(p-1)/2)/n)}{n \normtwo{d}\normtwo{d'}} \\
& = \frac{\sum_{j,k \in \Zn} d_j d'_k \delta_{k,j+(p-1)/2}}{\normtwo{d}\normtwo{d'}} \\
& = \frac{\sum_{j \in \Zn} d_j d'_{j+(p-1)/2}}{\normtwo{d}\normtwo{d'}} .
\end{align*}
If $2 n \mid p-1$, then $(p-1)/2 \equiv 0 \pmod{n}$, but if $2 n\nmid p-1$, then (since $n\mid p-1$ and $p$ is odd) we know that $n$ must be even and $p-1$ must be an odd multiple of $n$, so that $(p-1)/2$ is an odd multiple of $n/2$, and so $(p-1)/2 \equiv n/2 \pmod{n}$.
Thus
\[
A = \begin{cases}
\frac{1}{\normtwo{d}\normtwo{d'}} \sum_{j \in \Zn} d_j d'_j & \text{if $(p-1)/n$ is even,}\\
\frac{1}{\normtwo{d}\normtwo{d'}} \sum_{j \in \Zn} d_j d'_{j+n/2} & \text{if $(p-1)/n$ is odd,}
\end{cases}
\]
and $V_{d,d',p}=|A|^2$, so our formula $V_{d,d',p}=V^{(\sigma)}_{d,d'}$ is verified.
This shows that $V_{d,d',p}$ is a nonnegative real number equal to the squared magnitude of $\ip{d}{e'}/(\normtwo{d}\normtwo{e'})$,
where $e'$ is obtained from $d'$ by conjugating its elements and possibly cyclically shifting it (neither of these operations changes the $l^2$ norm).
So we can use the Cauchy--Schwarz inequality to see that $V_{d,d',p} \leq 1$.

We have $\PC_{\accentset{\circ}{f},\accentset{\circ}{f}}(0)=(p-1) \normtwosq{\ft{d}}$ and $\PC_{\accentset{\circ}{g},\accentset{\circ}{g}}(0)=(p-1)\normtwosq{\ft{d'}}$ by \cref{Paul}.
Thus
\begin{align*}
S_{d,d',p}
& = - 1 - U_{d,d'}-V_{d,d',p} + \frac{1}{p(p-1) \normtwosq{\ft{d}}\normtwosq{\ft{d'}}} \sum_{a \in \Fp} |\PC_{\accentset{\circ}{f},\accentset{\circ}{g}}(a)|^2  \\
& = - 1 - U_{d,d'} - V_{d,d',p} + \frac{p-1}{p} \PCDF(\accentset{\circ}{f},\accentset{\circ}{g}).\qedhere
\end{align*}
\end{proof}
Recall the function $\Phi$ defined in \eqref{Philip} in the Introduction.
We also define $\Xi\colon \R\to \R$ by $\Xi(x)=\Phi(x)+2/3$.
This $\Xi(x)$ is the same as $\Omega(1,x)$ defined in \cite[p.~6165]{Boothby-Katz}.
Note that both $\Phi$ and $\Xi$ are even continuous functions that are strictly decreasing on the interval $[0,1/2]$ and strictly increasing on $[1/2,1]$.
We now prove a technical lemma about $\Phi$ and $\Xi$ that we shall use in our approximation of crosscorrelation demerit factors.
\begin{lemma}\label{Laura}
For $x \in \F$, we have $|\Phi(x)| \leq 1/3$ and $|\Xi(x)| \leq 1$.
For $x,y \in \R$, we have $|\Phi(x)-\Phi(y)|=|\Xi(x)-\Xi(y)| \leq 2 |x-y|$.
\end{lemma}
\begin{proof}
The first claim is clear since $\Phi(x)$ has period $1$ and on $[0,1]$, it is quadratic with minimum value of $-1/6$ and maximum value $1/3$, while $\Xi(x)=\Phi(x)+2/3$ can be as small as $1/2$ or as large as $1$.

For the second claim, it is clear that $\Phi(x)-\Phi(y)=\Xi(x)-\Xi(y)$ since $\Xi(z)=\Phi(z)+2/3$ for every $z \in \R$; so we only need to prove the inequality $|\Phi(x)-\Phi(y)|\leq 2|x-y|$.
Let $x, y \in R$, and let $X=\{\sigma x+n: \sigma\in\{-1,1\}, n \in \Z\}$.
By the periodicity, evenness, and strict monotonicity of $\Phi$ on $[0,1/2]$, this $X$ is the set of all $x' \in \R$ with $\Phi(x')=\Phi(x)$, and $X$ contains precisely one point from $[0,1/2]$.
We similarly define $Y=\{\sigma y+n: \sigma\in\{-1,1\}, n \in \Z\}$, the set of all $y' \in \R$ with $\Phi(y')=\Phi(y)$.
The smallest distance between any point in $X$ and any point in $Y$ is equal to the distance between the unique point $x_0 \in X \cap [0,1/2]$ and the unique point $y_0 \in Y \cap[0,1/2]$, so $|x_0-y_0| \leq |x-y|$ and $\Phi(x_0)-\Phi(y_0)=\Phi(x)-\Phi(y)$.
So it suffices to prove $|\Phi(x_0)-\Phi(y_0)|\leq 2 |x_0-y_0|$.
Since $x_0,y_0 \in [0,1]$, we have
\begin{align*}
\Phi(x_0)-\Phi(y_0)
& = 2\left(x_0-\frac{1}{2}\right)^2 -\frac{1}{6} - 2\left(y_0-\frac{1}{2}\right)^2 +\frac{1}{6} \\
& = 2(x_0-y_0)(x_0+y_0-1),
\end{align*}
so that $|\Phi(x_0)-\Phi(y_0)|=2 |x_0+y_0-1| \cdot |x_0-y_0|$, and since $x_0,y_0 \in [0,1/2]$, we have $|x_0+y_0-1| \leq 1$.
\end{proof}
Now we present an approximation of the crosscorrelation demerit factor for instances of cyclotomic patterns.
\begin{proposition}\label{Adam}
Let $n \in \Z^+$, let $d$ and $d'$ be nonzero balanced cyclotomic patterns of index $n$, let $p$ be a prime with $p\equiv 1 \pmod{n}$, let $\accentset{\circ}{f}$ and $\accentset{\circ}{g}$ be the periodic sequences of length $p$ derived from $d$ and $d'$, respectively, let $r,r'\in\Z$, and let $f=\accentset{\circ}{f}^{(r)}$ and $g=\accentset{\circ}{g}^{(r')}$.
Then there is some $E \in \R$ with
\[
|E| \leq 192 \left(\frac{\normone{\ft{d}} \normone{\ft{d'}}}{\normtwo{\ft{d}}\normtwo{\ft{d'}}}\right)^2 \frac{\sqrt{p} (1+\log p)^3}{p-1} + \frac{22 p-4}{3 p (p-1)}.
\]
such that
\begin{multline*}
\left(\frac{p-1}{p}\right) \CDF(f,g)
= \frac{1}{3} +\frac{2}{3} \PCDF(\accentset{\circ}{f},\accentset{\circ}{g}) + U_{d,d'}\,\, \Phi\left(\frac{r-r'}{p}\right) \\ + V_{d,d',p} \,\, \Phi\left(\frac{r+r'}{p}\right) + E.
\end{multline*}
\end{proposition}
\begin{proof}
Note that if we scale $d$ (resp., $d'$), then $\ft{d}$, $\accentset{\circ}{f}$, and $f$ (resp, $\ft{d'}$, $\accentset{\circ}{g}$, and $g$) scale correspondingly, but $\CDF(f,g)$, $\PCDF(f,g)$, $U_{d,d'}$, $V_{d,d',p}$ and the upper bound for $E$ in the statement of this proposition are invariant to that scaling.
Thus, without loss of generality, we may assume that $d$ and $d'$ are scaled so that $\normtwo{\ft{d}}=\normtwo{\ft{d'}}=1$, which is one condition that Boothby and Katz insist upon on p.~6163 of \cite{Boothby-Katz}.
The other condition they insist upon on p.~6164 is that $\ft{d}_0=\ft{d'}_0=0$, which by \cref{Greta} is equivalent to our assumption in the statement of this proposition that $d$ and $d'$ are balanced.
Then one follows the proof of Theorem 1 in \cite{Boothby-Katz} from the beginning up to the formula
\[
\frac{\sum_{j \in \Z} |\AC_{f,g}(j)|^2}{\ell^2} = N_S + N_T + N_U + N_V + E_1,
\]
with the definitions of the quantities on the right hand side following immediately after.
We are in the special case $\ell=p$, and our advancement $r'$ is called $s$ in the paper, so that when we translate the formulas into the notation of this paper, we have
\begin{equation}\label{Vincent}
\begin{aligned}
N_S & = S_{d,d',p} \left(\frac{2 p^2+1}{3 p^2}\right), \\
N_T & = \Xi(0)=1, \\
N_U & = U_{d,d'} \,\, \Xi\left(\frac{r-r'}{p}\right), \\
N_V & = V_{d,d',p} \,\, \Xi\left(1+\frac{r+r'-1}{p}\right) = V_{d,d',p} \,\, \Xi\left(\frac{r+r'-1}{p}\right)\text{, and} \\
|E_1| & \leq 192 W_d^2 W_{d'}^2 \frac{(1+\log p)^3}{\sqrt{p}},
\end{aligned}
\end{equation}
where we have used our definition of $\Xi$ to provide the second equalities for $N_T$ and $N_V$.
Since $\ell=p$, we have
\[
\frac{\sum_{j \in \Z} |\AC_{f,g}(j)|^2}{\ell^2}=\frac{\sum_{j \in \Z} |\AC_{f,g}(j)|^2}{p^2}= N_S + N_T + N_U + N_V + E_1.
\]
By \cref{Abigail} and our assumptions about $d$ and $d'$, we have $\AC_{f,f}(0)=\AC_{g,g}(0)=p-1$, so that our equation becomes $(p-1)^2 \CDF(f,g)/p^2 = N_S + N_T + N_U + N_V + E_1$.
We substitute the values from \eqref{Vincent} to obtain
\begin{equation}\label{Irene}
\begin{aligned}
\left(\frac{p-1}{p}\right)^2 \CDF(f,g)
= \left(\frac{2}{3}+\frac{1}{3 p^2}\right) S_{d,d',p} + 1 & \\ + U_{d,d'} \,\, \Xi\left(\frac{r-r'}{p}\right) + V_{d,d',p} \,\, & \Xi\left(\frac{r+r'-1}{p}\right) + E_1.
\end{aligned}
\end{equation}
Recall that $W_d=\normone{\ft{d}}/\normtwo{\ft{d}}$ and $W_{d'}=\normone{\ft{d}}/\normtwo{\ft{d'}}$, so that
\[
|E_1| \leq 192 \left(\frac{\normone{\ft{d}} \normone{\ft{d'}}}{\normtwo{\ft{d}}\normtwo{\ft{d'}}}\right)^2 \frac{(1+\log p)^3}{\sqrt{p}}.
\]
Now we use the formula for $S_{d,d',p}$ from \cref{Matilda}\eqref{Annette} in \eqref{Irene} to obtain
\begin{multline*}
\left(\frac{p-1}{p}\right)^2 \CDF(f,g)
\\
= \left(\frac{2}{3}+\frac{1}{3 p^2}\right) \left(\frac{p-1}{p} \PCDF(\accentset{\circ}{f},\accentset{\circ}{g})-1-U_{d,d'}-V_{d,d',p}\right)
\\+ 1 + U_{d,d'} \,\, \Xi\left(\frac{r-r'}{p}\right) + V_{d,d',p} \,\, \Xi\left(\frac{r+r'-1}{p}\right) + E_1,
\end{multline*}
and if we multiply both sides by $p/(p-1)$ and rearrange, we get
\begin{multline*}
\left(\frac{p-1}{p}\right) \CDF(f,g)
= \left(\frac{2}{3}+\frac{1}{3 p^2}\right) \PCDF(\accentset{\circ}{f},\accentset{\circ}{g}) \\
-\left(\frac{2}{3}+\frac{2 p+1}{3 p(p-1)}\right) \left(1+U_{d,d'}+V_{d,d',p}\right) \\
+ \left(1+\frac{1}{p-1}\right) \left[1 + U_{d,d'} \,\, \Xi\left(\frac{r-r'}{p}\right) + V_{d,d',p} \,\, \Xi\left(\frac{r+r'-1}{p}\right) + E_1\right],
\end{multline*}
so that
\begin{equation}\label{Boris}
\begin{aligned}
\left(\frac{p-1}{p}\right) \CDF(f,g)
= \frac{2}{3} \PCDF(\accentset{\circ}{f},\accentset{\circ}{g}) & -\frac{2}{3} \left(1+U_{d,d'}+V_{d,d',p}\right) 
\\ + 1 + U_{d,d'} \,\, \Xi\left(\frac{r-r'}{p}\right) & + V_{d,d',p} \,\, \Xi\left(\frac{r+r'-1}{p}\right) + E_2,
\end{aligned}
\end{equation}
where
\begin{multline*}
E_2 = \left(\frac{p}{p-1}\right) E_1 + \frac{1}{3 p^2} \PCDF(\accentset{\circ}{f},\accentset{\circ}{g}) -\frac{2 p+1}{3 p(p-1)} \left(1+U_{d,d'}+V_{d,d',p}\right)  \\
+\frac{1}{p-1} \left[1 + U_{d,d'} \,\, \Xi\left(\frac{r-r'}{p}\right) + V_{d,d',p} \,\, \Xi\left(\frac{r+r'-1}{p}\right)\right].
\end{multline*}
From \cref{Matilda}\eqref{John}--\eqref{Gordon}, we know that $|U_{d,d'}|,|V_{d,d',p}| \leq 1$ for all $d,d'$, \cref{Raymond} shows that $\PCDF(\accentset{\circ}{f},\accentset{\circ}{g})\leq p$, and \cref{Laura} shows that $|\Xi(x)| \leq 1$ for all $x$, so we have
\[
|E_2| \leq 192 \left(\frac{\normone{\ft{d}} \normone{\ft{d'}}}{\normtwo{\ft{d}}\normtwo{\ft{d'}}}\right)^2 \frac{\sqrt{p} (1+\log p)^3}{p-1} + \frac{16 p+2}{3 p(p-1)}.
\]
We rearrange \eqref{Boris} to obtain
\begin{multline*}
\left(\frac{p-1}{p}\right) \CDF(f,g)
= \frac{1}{3} +\frac{2}{3} \PCDF(\accentset{\circ}{f},\accentset{\circ}{g}) + U_{d,d'}\,\, \Phi\left(\frac{r-r'}{p}\right) \\ + V_{d,d',p} \,\, \Phi\left(\frac{r+r'-1}{p}\right) + E_2.
\end{multline*}
Then we use \cref{Laura} to see that $\Phi((r+r'-1)/p)$ differs from $\Phi((r+r')/p)$ by no more than $2/p$, and since $|V_{d,d',p}| \leq 1$, we have
\begin{equation}\label{Gladys}
\begin{aligned}
\left(\frac{p-1}{p}\right) \CDF(f,g)
= \frac{1}{3} +\frac{2}{3} \PCDF(\accentset{\circ}{f},\accentset{\circ}{g}) & \\ + U_{d,d'}\,\, \Phi\left(\frac{r-r'}{p}\right) & + V_{d,d',p} \,\, \Phi\left(\frac{r+r'}{p}\right) + E,
\end{aligned}
\end{equation}
for some $E$ with
\[
|E| \leq 192 \left(\frac{\normone{\ft{d}} \normone{\ft{d'}}}{\normtwo{\ft{d}}\normtwo{\ft{d'}}}\right)^2 \frac{\sqrt{p} (1+\log p)^3}{p-1} + \frac{22 p-4}{3 p(p-1)}.
\]
It should be noted that $E$ must be real since all the other terms in \eqref{Gladys} are real.
\end{proof}
Now we generalize \cref{Adam} to codebooks.
\begin{proposition}\label{Aaron}
Let $n \in \Z^+$ and let $D$ be a balanced cyclotomic plan of index $n$ with $0\not\in D$.
Let $p$ be a prime with $p\equiv 1 \pmod{n}$, and for each $d \in D$, let $f_d$ be the periodic sequence of length $p$ from $d$.
Let $F=\{f_d: d \in D\}$, the periodic codebook of length $p$ from $D$.
For each $d \in D$, let $r_d \in \Z$, and let $g_d=f_d^{(r_d)}$ so that the codebook $G=\{g_d: d \in D\}$ is a $p$-instance of $D$.
If $|G|=|D|$, then there is some $E \in \R$ with
\[
|E| \leq 192 \left(\frac{1}{|D|} \sum_{d \in D} \frac{\normonesq{\ft{d}}}{\normtwosq{\ft{d}}}\right)^2 \frac{\sqrt{p} (1+\log p)^3}{p-1} + \frac{22 p-4}{3 p (p-1)}
\]
such that
\begin{multline*}
\left(\frac{p-1}{p}\right) \CDF(G)
= \frac{1}{3} +\frac{2}{3} \PCDF(F) + \frac{1}{|D|^2} \sum_{d,d' \in D} U_{d,d'}\,\, \Phi\left(\frac{r_d-r_{d'}}{p}\right) \\ + \frac{1}{|D|^2} \sum_{d,d' \in D} V_{d,d',p} \,\, \Phi\left(\frac{r_d+r_{d'}}{p}\right) + E.
\end{multline*}
\end{proposition}
\begin{proof}
For $d,d' \in D$, \cref{Adam} tells us that there is some $E_{d,d'} \in \R$ with
\[
|E_{d,d'}| \leq 192 \left(\frac{\normone{\ft{d}} \normone{\ft{d'}}}{\normtwo{\ft{d}}\normtwo{\ft{d'}}}\right)^2 \frac{\sqrt{p} (1+\log p)^3}{p-1} + \frac{22 p-4}{3 p (p-1)},
\]
such that
\begin{multline*}
\left(\frac{p-1}{p}\right) \CDF(g_d,g_{d'})
= \frac{1}{3} +\frac{2}{3} \PCDF(f_d,f_{d'}) + U_{d,d'}\,\, \Phi\left(\frac{r_d-r_{d'}}{p}\right) \\ + V_{d,d',p} \,\, \Phi\left(\frac{r_d+r_{d'}}{p}\right) + E_{d,d'},
\end{multline*}
and then one averages this expression over all $(d,d') \in D\times D$ to get the desired result.
\end{proof}

\subsection{Aperiodic correlation of codebooks from Hadamard plans}\label{Madeleine}

In this section we specialize the results of the previous section to codebooks derived from Hadamard plans.
\begin{lemma}\label{Maynard}
Let $n \in \Z^+$ and let $D$ be a Hadamard plan of index $n$.
\begin{enumerate}[(i).]
\item\label{Arthur} For $d,d' \in D$, we have
\[
U_{d,d'}= \begin{cases} 1 & \text{if $d=d'$,} \\
0 & \text{otherwise.}
\end{cases}
\]
\item\label{Abel} For $\sigma \in \Z/2\Z$, we have
\[
\frac{1}{|D|^2} \sum_{d,d' \in D} V^{(\sigma)}_{d,d'}=\frac{1}{n-1},
\]
and so for any prime $p$ with $p\equiv 1\pmod{n}$, we have
\[
\frac{1}{|D|^2} \sum_{d,d' \in D} V_{d,d',p}=\frac{1}{n-1}.
\]
\end{enumerate}
\end{lemma}
\begin{proof}
Since $U_{d,d'}=  \left|\ip{d}{d'}\right|^2/(\normtwosq{d}\normtwosq{d'})$ by \cref{Matilda}\eqref{John}, and since the elements of $D$ are orthogonal to each other, part \eqref{Arthur} follows.

Let $\sigma\in\Z/2\Z$.
Let $E=\{d/\normtwo{d}: d \in D\}$, which is an orthonormal set of balanced periodic sequences of length $n$.
Let $w$ be the periodic sequence $n^{-1/2} (1,1,\ldots,1)$ of length $n$, and let $G=E\cup\{w\}$, so that $G$ is an orthonormal basis of the $n$-dimensional $\C$-inner product space of periodic sequences of length $n$.
If $e=(e_0,e_1,\ldots,e_{n-1})$ is a periodic sequence of length $n$, let
\[
\tilde{e} = \begin{cases}
(\conj{e_0},\conj{e_1},\ldots,\conj{e_{n-1}}) & \text{if $\sigma=0$,} \\
(\conj{e_{n/2}},\conj{e_{n/2+1}},\ldots,\conj{e_{3 n/2-1}}) & \text{if $\sigma=1$,}
\end{cases}
\]
and recall that these indices are read modulo $n$, so the pattern in the second case is obtained from $e$ by cyclically shifting half way around and then conjugating all the elements.
Note that $\normtwo{\tilde{e}}=\normtwo{e}$.
Let $\tilde{E}=\{\tilde{e}: e \in E\}$, which is a set of balanced unit vectors.
We have
\begin{align*}
\frac{1}{|D|^2} \sum_{d,d' \in D} V^{(\sigma)}_{d,d'}
& = \frac{1}{|D|^2} \sum_{d,d' \in D} \frac{\left|\sum_{j \in \Zn} d_j d'_{j+\sigma(n/2)} \right|^2}{\normtwosq{d} \normtwosq{d'}} \\
& = \frac{1}{|D|^2} \sum_{e,e' \in E} \left|\sum_{j \in \Zn} e_j e'_{j+\sigma(n/2)} \right|^2 \\
& = \frac{1}{|D|^2} \sum_{e,e' \in E} |\ip{e}{\tilde{e'}}|^2 \\
& = \frac{1}{|D|^2} \sum_{e' \in E} \sum_{e \in G} |\ip{e}{\tilde{e'}}|^2,
\end{align*}
where we have used our definition of $E$ in the second equality, our definition of $\tilde{e'}$ in the third, and in the fourth the fact that $G=E\cup\{w\}$ with $w$ orthogonal to every $\tilde{e'}$ for $e' \in E$, so extending the summation of $e$ does not change the value.
Since $G$ is an orthonormal basis of the space of periodic sequences of length $n$, we know that for any $e' \in E$, we have $\sum_{e \in G} \left|\ip{e}{\tilde{e'}}\right|^2=\normtwosq{\tilde{e'}}$, which equals $1$.
Thus, the inner sum in our last expression is always $1$, and since $|E|=|D|=n-1$, we obtain the desired average value of $V^{(\sigma)}_{d,d'}$.  The average value of $V^{d,d',p}$ then follows immediately from \cref{Matilda}\eqref{Gordon}.
\end{proof}
Now we consider families of $p$-instances (or unimodularized $p$-instances) of a Hadamard plan.
\begin{theorem}\label{Victor}
Let $n \in \Z^+$ and let $D$ be a Hadamard plan of index $n$.
Let $\{A_\iota\}_{\iota \in A}$ be a family of instances (or a family of unimodularized instances) of $D$.
For each $\iota \in I$, let $p_\iota$ be a prime with $p_\iota \equiv 1 \pmod{n}$ and let $r_\iota \colon D \to \Z$ be a function such that $A_\iota=\{f_{\iota,d}d^{(r_\iota(d))}: d \in D\}$, where $f_{\iota,d}$ is the periodic sequence of length $p_\iota$ from $d$ (or a unimodularized periodic sequence of length $p_\iota$ from $d$).
Let $\rho\colon D \to \R$ be a function and let $\sigma \in \Z/2\Z$.
Suppose that $\{p_\iota\}_{\iota \in I}$ is unbounded, that $(p_\iota-1)/n \equiv \sigma \pmod{2}$ for every $\iota \in I$, and that for every $d \in D$, the quantity $r_\iota(d)/p_\iota$ tends to $\rho(d)$ as $p_\iota$ tends to infinity.
Then $\CDF(A_\iota)$ tends to
\[
1+ \frac{1}{3(n-1)} + \frac{1}{(n-1)^2} \sum_{d,d' \in D} V^{(\sigma)}_{d,d'} \,\, \Phi\left(\rho(d)+\rho(d')\right).
\]
as $p_\iota$ tends to infinity.
\end{theorem}
\begin{proof}
First we prove the result when $\{A_\iota\}_{\iota\in I}$ is a family of instances of $D$.
By \cref{Nora}, there is some $N$ such that whenever $p_\iota \geq N$, we have $|A_\iota|=|D|$, and we set $J=\{\iota \in I: p_\iota \geq N\}$.
Since our result is about a limit as $p_\iota \to \infty$, it suffices to prove our result with $\{A_\iota\}_{\iota \in J}$ in place of $\{A_\iota\}_{\iota \in I}$.
For each $\iota \in J$, let $P_\iota$ be the periodic codebook $\{f_{\iota,d}: d \in D\}$.
Since $D$ is Hadamard, $\normtwo{\ft{d}}=1$ for every $d \in D$ and $|D|=n-1$, and so \cref{Aaron} says that there is some $E_\iota \in \R$ with
\[
|E_\iota| \leq 192 \left(\frac{1}{n-1} \sum_{d \in D} \normonesq{\ft{d}}\right)^2 \frac{\sqrt{p_\iota} (1+\log p_\iota)^3}{p_\iota-1} + \frac{22 p_\iota-4}{3 p_\iota (p_\iota-1)}
\]
such that
\begin{multline*}
\left(\frac{p_\iota-1}{p_\iota}\right) \CDF(A_\iota)
= \frac{1}{3} +\frac{2}{3} \PCDF(P_\iota) \\ + \frac{1}{(n-1)^2} \sum_{d,d' \in D} U_{d,d'}\,\, \Phi\left(\frac{r_{\iota}(d)-r_{\iota}(d')}{p_\iota}\right) \\ + \frac{1}{(n-1)^2} \sum_{d,d' \in D} V_{d,d'}^{(\sigma)} \,\, \Phi\left(\frac{r_{\iota}(d)+r_{\iota}(d')}{p_\iota}\right) + E_\iota,
\end{multline*}
where we have used \cref{Matilda}\eqref{Gordon} to replace $V_{d,d',p_\iota}$ with $V^{(\sigma)}_{d,d'}$ because of our assumption on the parity of $(p_\iota-1)/n$.
Using the values of $U_{d,d'}$ from \cref{Maynard}\eqref{Arthur} and of $\PCDF(P_\iota)$ from \cref{Heloise}, we obtain
\begin{multline*}
\left(\frac{p_\iota-1}{p_\iota}\right) \CDF(A_\iota)
= \frac{1}{3} +\frac{2}{3} \left(\frac{p_\iota}{p_\iota-1}\right) + \frac{1}{(n-1)^2} \sum_{d \in D} \Phi\left(\frac{r_{\iota}(d)-r_{\iota}(d)}{p_\iota}\right) \\ + \frac{1}{(n-1)^2} \sum_{d,d' \in D} V_{d,d'}^{(\sigma)} \,\, \Phi\left(\frac{r_{\iota}(d)+r_{\iota}(d')}{p_\iota}\right) + E_\iota,
\end{multline*}
which equals
\begin{multline*}
\left(\frac{p_\iota-1}{p_\iota}\right) \CDF(A_\iota)
= \frac{1}{3} +\frac{2}{3} \left(\frac{p_\iota}{p_\iota-1}\right) + \frac{1}{3 (n-1)}\\ + \frac{1}{(n-1)^2} \sum_{d,d' \in D} V_{d,d'}^{(\sigma)} \,\, \Phi\left(\frac{r_{\iota}(d)+r_{\iota}(d')}{p_\iota}\right) + E_\iota,
\end{multline*}
since $\Phi(0)=1/3$.
As $p_\iota$ tends to infinity, $E_\iota$ tends to zero, $p_\iota/(p_\iota-1)$ tends to $1$, and by continuity of $\Phi$, each term $\Phi((r_{\iota}(d)+r_{\iota}(d'))/p_\iota)$ tends to $\Phi(\rho(d)+\rho(d'))$.
Thus, $\CDF(A_\iota)$ tends to the stated limit.

When $\{A_\iota\}_{\iota \in I}$ is a family of unimodularized instances of $D$, then for each $\iota \in I$ and $d \in D$, let $g_{\iota,d}$ be the periodic sequence of length $p_\iota$ from $d$, so that the sequence $f_{\iota,d}$ in the statement of this theorem is unimodularization of $g_{\iota,d}$.
For each $\iota \in J$, let $P_\iota=\{g_{\iota,d}: d \in D\}$ and $\uni{P}_\iota=\{f_{\iota,d}: d \in D\}$, so that $\uni{P}_\iota$ is a unimodularization of $P_\iota$, and $A_\iota$ is a rotation of $\uni{P}_\iota$.
Let $B_\iota=\{g_{\iota,d}^{(r_\iota(d))}: d \in D\}$, so that the part of this theorem we have already proved shows that, as $p_\iota$ tends to infinity, $\CDF(B_\iota)$ tends to the limit we desire to show for $\CDF(A_\iota)$.
Note that $\len(P_\iota)=p_\iota$ and $|P_\iota|=|\uni{P}_\iota|=|D|$ for every $\iota \in I$.
By \cref{Nora}, there is some $M$ such that whenever $p_\iota \geq M$, we have $|A_\iota|=|B_\iota|=|D|$, and we set $K=\{\iota \in I: p_\iota \geq M\}$.
If we let $\{P_\iota\}_{\iota \in K}$, $\{\uni{P}_\iota\}_{\iota \in K}$, $\{B_\iota\}_{\iota \in K}$, and $\{A_\iota\}_{\iota \in K}$ respectively serve the roles of $\{P_\iota\}_{\iota \in I}$, $\{\uni{P}_\iota\}_{\iota \in I}$, $\{A_\iota\}_{\iota \in I}$, and $\{\uni{A}_\iota\}_{\iota \in I}$  in \cref{Karl}, we see that $\CDF(A_\iota)$ has the same limiting behavior that $\CDF(B_\iota)$ does as $p_\iota$ tends to infinity.
\end{proof}
Recall the definition of a coherently rotated family from \cref{Rudolph} in the Introduction.
When we specialize \cref{Victor} to such families, we get an especially tidy limit, which we presented in the Introduction as \cref{Mary}.
The following corollary of \cref{Victor} is more general than \cref{Mary}, in that it works for both families of instances and families of unimodularized instances of the plan.
When we restrict to unimodularized instances (as in \cref{Mary}), the limit $\SDC(A_\iota) \to \infty$ is the same as $\len(A_\iota) \to \infty$ because $\SDC(A_\iota)=\len(A_\iota)$ when $A_\iota$ is a uniform-length, unimodular codebook.
\begin{corollary}\label{Julia}
Let $n \in \Z^+$, let $D$ be a Hadamard plan of index $n$, and let $\rho \in \R$.
Let $\{A_\iota\}_{\iota \in A}$ be a coherently $\rho$-rotated family of instances (or a coherently $\rho$-rotated family of unimodularized instances) of $D$.
Then $\CDF(A_\iota)$ tends to
\[
1+ \frac{1}{3(n-1)} + \frac{1}{n-1} \Phi\left(2 \rho\right).
\]
as $\SDC(A_\iota)$ tends to infinity.
For a fixed $n$, the right hand side achieves its global minimum value of $1+1/(6(n-1))$ if and only if $\rho \in \{(2n+1)/4: n \in \Z\}$.
\end{corollary}
\begin{proof}
If each $A_\iota$ is a $p_\iota$-instance of $D$, then $\AC_{f,f}(0)=p_\iota-1$ for every $f \in A_\iota$ by \cref{Abigail}, and so $\SDC(A_\iota)=p_\iota-1$.
If each $A_\iota$ is a unimodularized $p_\iota$-instance of $D$, then $\AC_{f,f}(0)=\len(f)=p_\iota$.
So for each $\iota\in I$, there is a well defined $p_\iota$ such that $A_\iota$ is a $p_\iota$ instance (or unimodularized $p_\iota$-instance) of $D$ and the limit $p_\iota \to\infty$ in \cref{Victor} is the same as the limit $\SDC(A_\iota)\to\infty$.

We first prove this theorem under the additional assumption that there is some $\sigma \in \Z/2\Z$ such that $(p_\iota-1)/n\equiv \sigma\pmod{2}$ for every $\iota \in I$.
As $\SDC(A_\iota)\to\infty$, we have $p_\iota \to \infty$ and may apply \cref{Victor} to see that $\CDF(A_\iota)$ tends to
\[
1+ \frac{1}{3(n-1)} + \frac{1}{(n-1)^2} \sum_{d,d' \in D} V^{(\sigma)}_{d,d'} \,\, \Phi(2\rho),
\]
which by \cref{Maynard}\eqref{Abel} equals
\[
1+ \frac{1}{3(n-1)} + \frac{1}{n-1} \Phi(2\rho).
\]
Note that the limiting $\CDF$ is independent of $\sigma$.

We use the lack of dependence on $\sigma$ in this last limit to relax the condition on the parity of $(p_\iota-1)/n$.
We partition $I$ into $I_0$ and $I_1$, with $I_\tau=\{\iota \in I: (p_\iota-1)/2 \equiv \tau\pmod{2}\}$ for each $\tau \in \Z/2\Z$.
If the set $S_\tau=\{p_\iota: \iota \in I_\tau\}$ is bounded for some $\tau \in \Z/2\Z$, then the limiting behavior of $\CDF(A_\iota)$ for $\{A_\iota: \iota \in I\}$ as $\SDC(A_\iota)\to\infty$ is precisely the same as that for $\{A_\iota:\iota \in I_{1-\tau}\}$, and so we may apply the special case we have already proved.
On the other hand, if both $S_0$ and $S_1$ are unbounded, then we note that $\CDF(A_\iota)$ tends to the same limit for both $\{A_\iota:\iota \in I_0\}$ and $\{A_\iota:\iota \in I_1\}$, so it tends to that limit for the entire family $\{A_\iota\}_{\iota \in I}$.

Since $\Phi$ is periodic with period $1$ and $\Phi(x)=2(x-1/2)^2-1/6$ for $0 \leq x \leq 1$, it is evident that $\Phi$ achieves its global minimum value of $-1/6$ for $x \in \{(2 n+1)/2: n \in \Z\}$, so our function is minimized when $\rho \in \{(2 n+1)/4: n \in \Z\}$, and the minimum value is $1+1/(6(n-1))$.
\end{proof}

\section*{Acknowledgement}
The authors thank Evgeniya Lagoda for many helpful discussions, and thank her especially for work with the second author that led to the bounds in \cref{Esmeralda}, \cref{Edna}, and \cref{Mildred}.

\end{document}